\documentclass[aps,pra,showpacs,twoside,twocolumn,10pt]{revtex4-2}
\usepackage[colorlinks=true, citecolor=red, urlcolor=blue ]{hyperref}
\usepackage{epsfig,newlfont,amssymb,amsfonts,amsmath,bm,subfigure,palatino,mathtools,amsthm,braket,soul,enumitem,color,graphics,graphicx,times,physics,bbold}
\usepackage[normalem]{ulem}
\usepackage{xcolor}
\usepackage{physics}
\usepackage{dsfont}
\usepackage{mathrsfs}
\usepackage{verbatim}
\usepackage{amsthm,amssymb}
\usepackage{comment}
\usepackage{bigints}
\usepackage{amsmath,amssymb,amsthm}

\newtheorem{theorem}{Theorem}
\newtheorem{lemma}{Lemma}

\begin{document}

\AtBeginDocument{
  \renewcommand{\Re}{\mathfrak{R}}
  \renewcommand{\Im}{\mathfrak{I}}
}

\title{ 
Exact Tradeoff Between Quantum Error Correction and Quantum Darwinism: An Information-Theoretic No-Go Theorem
}

\author{Arghya Maity}
\affiliation{School of Physical and Mathematical Sciences, Nanyang Technological University, 21 Nanyang Link, Singapore 637371, Singapore}

\author{Kelvin Onggadinata}
\affiliation{School of Physical and Mathematical Sciences, Nanyang Technological University, 21 Nanyang Link, Singapore 637371, Singapore}

\author{Teck Seng Koh}
\affiliation{School of Physical and Mathematical Sciences, Nanyang Technological University, 21 Nanyang Link, Singapore 637371, Singapore}

\begin{abstract}
Quantum error correction (QEC) and Quantum Darwinism describe opposing consequences of system--environment interactions: QEC seeks to preserve logical quantum information, whereas Quantum Darwinism explains the emergence of objective classical information through its proliferation into the environment. Despite their common physical origin, no direct quantitative connection between these paradigms has previously been established. 
We introduce an exactly solvable block--environment model based on the logical GHZ block of the Shor [[9,1,3]] quantum error-correcting code, collectively coupled to N environment qubits.
The logical fidelity, Holevo information, and Darwinistic redundancy are obtained systematically for arbitrary environment size and imperfect recovery efficiency. Eliminating the common decoherence parameter yields an exact tradeoff relating Darwinistic redundancy directly to the post-recovery logical fidelity, demonstrating that the emergence of redundant classical records occurs at the expense of logical quantum information. We further prove a model-independent no-go theorem showing that the logical fidelity exceeds a critical threshold precludes the emergence of Darwinistic redundancy, irrespective of the microscopic Hamiltonian or environment structure. The solvable model saturates this general bound, establishing the first quantitative information-theoretic connection between logical quantum information protection and the emergence of redundant classical records.
\end{abstract}

\maketitle
\textbf{\textit{Introduction}:}~Quantum
information can either remain protected within a quantum
system or become encoded in its surrounding environment,
but it cannot do both simultaneously.
This fundamental competition lies at the heart of two major
developments in quantum information science: quantum error
correction~\cite{Shor_PRA_1995, Steane_PRL_1996,Shor_PRA_1996,Zurek_PRL_1996,Knill_PRA_1997,Gottesman_Thesis_1997,Shor_IEEE_1998}, which
protects quantum coherence against environmental noise, and
quantum Darwinism~\cite{Zurek_RMP_2003,Zurek_PRA_2005, Zurek_NP_2009, Zurek_PRL_2009, Jelezko_PRL_2019, Paternostro_PRA_2018, Pan_Sci-Bulletin_2019},
which explains the emergence of objective classical reality
through the redundant encoding of information in the
environment. Despite their common origin in system--environment
interactions, a direct quantitative connection between these
two paradigms has not been established.
Such a connection would reveal the fundamental limits
governing both the protection of quantum coherence and the
formation of classical objectivity, and would place the
quantum-to-classical transition on quantitative footing.

The absence of such a quantitative connection is closely
related to the lack of a common analytical framework in which
both quantum error correction and Quantum Darwinism can be
treated on equal footing. While QEC is naturally characterized
by logical fidelities after recovery, Quantum Darwinism is
quantified through the Holevo information and the redundancy
of environment fragments. Bridging these paradigms therefore
requires a model in which both quantities can be obtained
analytically from the same microscopic dynamics. A very recent study explored connections between Quantum
Darwinism and quantum error correction from an
operator-algebraic perspective~\cite{Girard_arXiv_2026}.
Complementary to this algebraic approach, we develop an
exactly solvable dynamical model in which both the logical
fidelity and the Darwinistic redundancy are obtained in
closed form. Eliminating the common decoherence parameter
yields an analytical tradeoff that is independent of the
microscopic dynamical parameters.
%microscopic interaction time, coupling strength, and environment size.

% --- P3: Model ---
In this work, we introduce a block--environment model in
which a logical qubit is encoded in one GHZ block of the
Shor $[[9,1,3]]$ code and interacts collectively with $N$
independent environment qubits through a non-commuting
Hamiltonian. Because each GHZ block experiences identical
block--environment dynamics in the present model, analyzing
a single block faithfully captures the information flow
throughout the full nine-qubit encoding. The commuting sector
of the model is exactly solvable, yielding closed-form
expressions for the logical fidelity, Holevo information,
and Darwinistic redundancy for arbitrary environment size
and imperfect recovery efficiency. The robustness of these
analytical results against the non-commuting interaction is
established by exact numerical simulations.

% --- P4: Results ---
Our central result is a closed-form tradeoff theorem showing
that the Darwinistic redundancy $R_\delta$ is uniquely
determined by the post-recovery logical fidelity
$F_L^{(N)}$, with the microscopic dynamical parameters
eliminated from the final relation. Any gain in classical
objectification therefore comes at a quantifiable cost to
quantum error-correction performance. Beyond the solvable
model, we prove a model-independent no-go theorem showing
that whenever the bare logical fidelity exceeds a threshold
set by the binary entropy, Darwinistic redundancy cannot
emerge irrespective of the microscopic Hamiltonian or
environment structure. Together, these results establish the
first quantitative information-theoretic framework linking
logical quantum information protection to the emergence of
redundant classical information.

% --- P5: Broader significance ---
These results also have a direct operational consequence:
syndrome-based recovery~\cite{Somaroo_PRL_1998,Pan_Nature_2012,Chow_Nature-Comm_2015,Schoelkopf_Nature_2016,Google_Nature_2023,IBM_Nature_2022,Anderson_PRX_2021} and environment-fragment readout~\cite{Jelezko_PRL_2019,Paternostro_PRA_2018, Pan_Sci-Bulletin_2019},
although experimentally distinct, become equivalent probes of the
same underlying decoherence process. Measuring the
post-recovery logical fidelity uniquely determines the
Darwinistic information encoded in the environment, and
vice versa. Because the solvable model saturates the
general no-go bound, it represents a extremal realization
of the competition between logical quantum information
protection and redundant classical information, providing a
general framework for exploring the transition from protected logical quantum information to redundant classical information.

% ============================================================
%Section II: Model and Exact Results
%============================================================
\textbf{\textit{Exactly Solvable Block--Environment Model.}}---
We consider the logical GHZ block of the Shor $[[9,1,3]]$
quantum error-correcting code interacting collectively with an
environment of $N$ independent qubits. 
Since each block--environment pair undergoes identical
dynamics, it is sufficient to focus on a representative
logical block $b$. The logical
basis is formed by the orthogonal codewords
$|\bar z_\pm\rangle_b=(|000\rangle\pm|111\rangle)/\sqrt2$,
while each environment qubit is initially prepared in
$|+\rangle=(|0\rangle+|1\rangle)/\sqrt2$.
The block interacts with the environment through the Hamiltonian
\begin{equation}
\label{eq:H}
\hat H
 =
g_Z\,\hat Z_b\otimes\hat S_Z
+
g_X\,\hat X_b\otimes\hat S_X;~~
\hat S_Z =\sum_{k=1}^{N}Z_k , ~ 
\hat S_X=\sum_{k=1}^{N}X_k,
\end{equation}
where
$\hat Z_b = Z^{\otimes3}$ and
$\hat X_b = X^{\otimes3}$
act on the three-qubit logical block, while
$\hat S_Z$ and $\hat S_X$ act collectively on the
$N$-qubit environment.
Within the logical subspace,
$\hat Z_b$ exchanges the two logical codewords $(\hat{Z}_b\ket{\bar{z}_\pm}=\ket{\bar{z}_\mp})$,
whereas $\hat X_b$ acts as the corresponding logical phase
operator $(\hat{X}_b\ket{\bar{z}_\pm}=\pm \ket{\bar{z}_\pm})$.
The analytical results are obtained in the exactly solvable
limit $
\hat H_0
=
g_Z\,\hat Z_b\otimes\hat S_Z,
$
corresponding to $g_X=0$ (or perturbatively
$g_X\ll g_Z$). The full Hamiltonian
Eq.~(\ref{eq:H}) is retained in the numerical analysis to test
the robustness of the analytical tradeoff beyond the exactly
solvable regime.

Since $\left[ \hat Z_b\otimes Z_i,\, \hat Z_b\otimes Z_j \right] = 0,$ the propagator generated by $H_0$ factorizes exactly,
$
\hat U_0(t) = e^{-it\hat H_0} = \prod_{k=1}^{N}
\exp\!\left(
-ig_Zt\,\hat Z_b\otimes Z_k \right).$
The detailed derivation is given in the Appendix. Starting from
$
|\Psi(0)\rangle
=
|\bar z_+\rangle_b
\otimes
|+\rangle^{\otimes N},
$
the exact evolved state admits the branch decomposition
\begin{equation}
|\Psi_{+}(t)\rangle
=
\tfrac{1}{2}
\bigl(|\bar{z}_+\rangle_b+|\bar{z}_-\rangle_b\bigr)
|\phi_+(t)\rangle
+
\tfrac{1}{2}
\bigl(|\bar{z}_+\rangle_b-|\bar{z}_-\rangle_b\bigr)
|\phi_-(t)\rangle,
\label{eq:psi}
\end{equation}
where the conditional environment states are
$
|\phi_\pm(t)\rangle
=
e^{\mp ig_Zt\hat S_Z}
|+\rangle^{\otimes N}.
$
The analytically solvable interaction is a pure-dephasing
model in the $\hat Z_b$ eigenbasis
$\{|000\rangle,|111\rangle\}$, analogous to standard
spin-environment dephasing models~\cite{Zurek_PRD_1982}.
Throughout this work, however, we describe the dynamics in
the logical basis
$\{|\bar z_+\rangle,|\bar z_-\rangle\}$,
where $\hat Z_b$ exchanges the logical codewords
($\hat Z_b|\bar z_\pm\rangle=|\bar z_\mp\rangle$),
so that the same interaction appears as logical-state
transitions.
Because the operators \(Z_k\) commute, these states admit the
product representation
\begin{equation}
|\phi_\pm(t)\rangle
=
\bigotimes_{k=1}^{N}
\left[
\cos(g_Zt)|+\rangle_k
\mp
i\sin(g_Zt)|-\rangle_k
\right].
\end{equation}
Tracing out the environment gives $\rho_b(t) = \mathrm{diag}\left[(1\!+\!\Gamma_N)/2,\,(1\!-\!\Gamma_N)/2\right]$
in the logical basis
\(\{|\bar z_+\rangle_b,|\bar z_-\rangle_b\}\), governed entirely by the decoherence factor
\begin{equation}
\Gamma_N(t)
=
\langle\phi_-(t)|\phi_+(t)\rangle
=
\cos^N(2g_Zt).
\label{eq:Gamma}
\end{equation}
Thus the complete reduced dynamics of the logical block is
governed by the single decoherence factor \(\Gamma_N(t)\).
\begin{figure*}[htb]
\includegraphics[width=\textwidth]{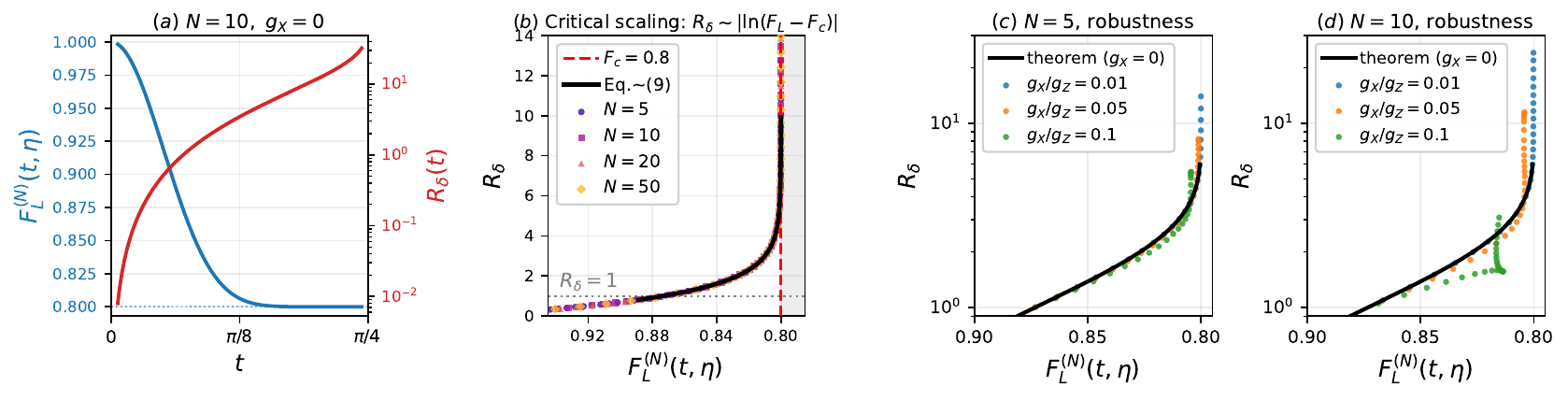}
\caption{
Numerical results for the Darwinism--QEC competition
($g_Z=1$, $\eta=0.6$, $\delta=0.10$).
\textit{(a)}~Competition in time for $N=10$, $g_X=0$.
Blue curve (left axis): logical fidelity $F_L^{(N)}(t,\eta)$,
decaying monotonically from unity toward the floor
$F_c=(1+\eta)/2=0.8$ (dotted line).
Red curve (right axis, log scale): Darwinistic redundancy
$R_\delta(t)$, rising simultaneously.
Solid lines are the closed-form analytical predictions
Eqs.~\eqref{eq:FL} and~\eqref{eq:R};
exact diagonalization data (not shown separately) agrees to these.
\textit{(b)}~Critical Darwinism--QEC scaling.
$R_\delta$ is plotted against the logical fidelity $F_L^{(N)}$
on linear axes for $N=5,10,20,50$ (colored points).
The black curve is the exact tradeoff
(Eq.~\eqref{eq:tradeoff}).
Data from all four environment sizes collapse onto the same
analytical curve, confirming that $N$ is absorbed entirely
into the fidelity gap $\varepsilon=F_L^{(N)}-F_c$.
The red dashed line marks the critical fidelity $F_c=0.8$;
as $F_L^{(N)}\to F_c^+$, the redundancy diverges
logarithmically ($R_\delta\sim|\ln\varepsilon|$), while for
$F_L^{(N)}>F_{\mathrm{th}}=0.874$ the redundancy is identically
zero (no Darwinism).
\textit{(c),(d)}~Robustness of the tradeoff beyond the
analytically solvable limit for $N=5$ and $N=10$.
Colored points show exact diagonalization results at
$g_X/g_Z=0.01,\,0.05,\,0.10$; the black curve is the
$g_X=0$ theorem Eq.~\eqref{eq:tradeoff}.
The tradeoff remains in excellent agreement with the
analytical prediction for $g_X/g_Z\lesssim0.05$, with
systematic deviations growing as $g_X$ increases.
}
\label{fig:competition}
\end{figure*}

\textit{Logical fidelity.}---
We model recovery via the
imperfect channel
$
\mathcal E_\eta(\rho)
=
P_+\rho P_+
+
\eta\,
\bigl(\hat Z_b\otimes\mathbb I_E\bigr)
P_-\rho P_-
\bigl(\hat Z_b\otimes\mathbb I_E\bigr)
+
(1-\eta)\,
P_-\rho P_-,
$
where
$
P_\pm
=
|\bar z_\pm\rangle_b
\langle\bar z_\pm|
\otimes
\mathbb I_E
$
are the syndrome projectors, and
$\eta\in[0,1]$ denotes the probability of successfully
correcting the $P_-$ branch. 
Tracing out the environment after applying $\mathcal{E}_\eta$
gives the exact post-recovery logical fidelity
\begin{equation}
F_L^{(N)}(t,\eta)
=
\frac{
1+\eta+(1-\eta)\Gamma_N(t)
}{2}.
\label{eq:FL}
\end{equation}
\textit{Holevo information.}---
The classical information accessible to an observer holding a
fragment of $m$ environment qubits is quantified by the
Holevo information
\begin{equation}
\chi_m
=
S(\bar\rho_F)
-
\frac12S(\rho_F^+)
-
\frac12S(\rho_F^-),
\end{equation}
where
$S(\rho)=-\mathrm{Tr}(\rho\ln\rho)$
is the von Neumann entropy,
$\rho_F^\pm$ are the fragment states conditioned on the two
logical codewords,
and
$\bar\rho_F=(\rho_F^++\rho_F^-)/2$
is their average.
Because $|\phi_\pm(t)\rangle$ are product states, the conditional
fragment states are pure ($S(\rho_F^\pm)=0$) and their overlap
factorizes as $\gamma_m(t)=\cos^m(2g_Zt)=[\Gamma_N(t)]^{m/N}$.
The Holevo information is therefore (Appendix)
\begin{equation}
\chi_m(t)
=
H_2\!\left(
\frac{1+\cos^m(2g_Zt)}{2}
\right),
\label{eq:chi}
\end{equation}
where $H_2(p)=-p\ln p-(1-p)\ln(1-p)$ is the binary entropy.
For any nontrivial interaction
($|\cos(2g_Zt)|<1$),
the overlap decreases exponentially with fragment size, so
$\chi_m$ increases monotonically toward its maximum value
$\ln2$, reflecting the increasing distinguishability of the
conditional fragment states and hence the progressive
accumulation of accessible classical information about the
logical state in the environment.

\textit{Darwinistic redundancy.}---Following
Refs.~\cite{Zurek_PRA_2005,Zurek_NP_2009,
Castro_PRL_2019,Korbicz_Quantum_2021},
the Darwinistic redundancy $R_\delta$ quantifies how many
independent observers can simultaneously infer the logical
state by accessing disjoint environment fragments.
Here $\delta\in(0,1)$ specifies the tolerated information
deficit relative to the maximum accessible classical
information $\ln2$.
The redundancy is defined as
$R_\delta=N/m_\delta$, where
$m_\delta=\min\{m:\chi_m\ge(1-\delta)\ln2\}$
is the minimum fragment size satisfying the Darwinism
criterion.
%The redundancy is defined as $R_\delta=N/m_\delta$, where $m_\delta$ is the minimum fragment size satisfying $\chi_{m_\delta}\ge(1-\delta)\ln2$.
Combining Eqs.~\eqref{eq:chi} and~\eqref{eq:Gamma} yields the
closed-form expression
\begin{equation}
R_\delta(t)
=
N\,
\frac{\ln[\cos(2g_Zt)]}
{\ln\gamma^*(\delta)},
\label{eq:R}
\end{equation}
where $\gamma^*(\delta)\in(0,1)$ is the threshold fragment
overlap defined implicitly by
$
H_2\!\left(\frac{1+\gamma^*}{2}\right)
=
(1-\delta)\ln2.
$
Equation~\eqref{eq:R} establishes the analytical dependence
of Darwinistic redundancy on the microscopic decoherence
process, providing the quantitative link between logical
quantum information preservation and the redundant
encoding of classical information in the environment.

\textit{Exact tradeoff theorem.}---
The logical fidelity and the Darwinistic redundancy are both
controlled by the same decoherence factor $\Gamma_N(t)$.
Eliminating $\Gamma_N$ between
Eqs.~\eqref{eq:FL} and~\eqref{eq:R}
yields the central result
\begin{equation}
R_\delta
=
\frac{1}{\ln\gamma^*(\delta)}
\ln\!\left[
\frac{2F_L^{(N)}-(1+\eta)}{1-\eta}
\right].
\label{eq:tradeoff}
\end{equation}
Equation~\eqref{eq:tradeoff} is independent of the microscopic
dynamical parameters: the interaction time, coupling strength,
and environment size are completely eliminated.
The Darwinistic redundancy is therefore determined solely by
the post-recovery logical fidelity $F_L^{(N)}$ and the recovery
efficiency $\eta$.
Since $\ln\gamma^*(\delta)<0$,
$R_\delta$ decreases monotonically with $F_L^{(N)}$,
showing that any gain in logical quantum information comes
at a quantifiable cost to the redundant encoding of classical
information.
The tradeoff depends only on the observable logical fidelity,
not on the microscopic parameter values that generated it.

Equation~\eqref{eq:tradeoff} also has a direct operational
interpretation.
Because both the logical fidelity and the Holevo information
are determined by the same decoherence factor,
syndrome-based recovery and environment-fragment readout
provide equivalent experimental probes of the same
system--environment information flow.
Experimentally, measuring the post-recovery logical fidelity
uniquely identifies the underlying decoherence factor.
Since the Holevo information of every environment fragment
depends on the same quantity, the entire fragment-information
curve can then be reconstructed without direct measurements
on the environment. Conversely, measurements of the
fragment-information curve determine the logical fidelity
through Eq.~\eqref{eq:tradeoff}. The tradeoff therefore
establishes a direct operational bridge between quantum error
correction and Quantum Darwinism, demonstrating that logical
quantum information preservation and redundant classical
information are complementary manifestations of the same
underlying system--environment information flow.

\textit{Logarithmic sensitivity near full objectification.}---
Defining the fidelity gap
$\varepsilon=F_L^{(N)}-F_c\ge0$,
which measures the distance from the minimum achievable
logical fidelity,
with
$F_c=(1+\eta)/2$,
the tradeoff becomes
\begin{equation}
R_\delta
=
-\frac{1}{|\ln\gamma^*(\delta)|}
\ln (F_L^{(N)}-F_c)
+\mathrm{const},
\label{eq:log_sensitivity}
\end{equation}
where the additive constant depends only on the recovery
efficiency $\eta$.
Equation~\eqref{eq:log_sensitivity} shows that the redundancy
grows logarithmically as the logical fidelity approaches
$F_c$, corresponding to complete objectification
($\Gamma_N\rightarrow0$).
Correspondingly, the sensitivity of the redundancy to changes
in the logical fidelity is
\begin{equation}
\left|
\frac{\partial R_\delta}
{\partial F_L^{(N)}}
\right|
=
\frac{1}
{|\ln\gamma^*(\delta)|\,\varepsilon},
\end{equation}
which diverges as $\varepsilon\rightarrow0^+$, showing that
near complete objectification an arbitrarily small decrease
in logical fidelity produces an increasingly large increase
in Darwinistic redundancy.
For finite environments this growth is naturally cut off by
the bound $R_\delta\le N$, while numerical data for different
environment sizes collapse onto the exact tradeoff curve
(Fig.~\ref{fig:competition}b), confirming that the tradeoff
depends only on the logical fidelity and not on the
microscopic dynamical parameters.
Equation~\eqref{eq:log_sensitivity} therefore identifies a
distinct logarithmic-sensitivity regime in which a small loss
of logical quantum information produces a disproportionately
large increase in redundant classical records.

% ============================================================
% PRL Section III: No-Go Theorem
% ============================================================

\textbf{\textit{No-go theorem for Darwinistic redundancy.}}---We now show that the competition revealed by the tradeoff theorem is not specific to the solvable model, but follows from a general information-theoretic bound.

% ---- Lemma ----
\begin{lemma}[Block entropy bound]
\label{lem:entropy_bound}
%\textbf{\textit{Lemma.}}---
For any qubit block state $\rho_B$ with
bare logical fidelity
$F_{\mathrm{bare}}=\langle\bar{z}_+|\rho_B|\bar{z}_+\rangle$,
the von Neumann entropy satisfies
\begin{equation}
S(\rho_B)
\le
H_2(F_{\mathrm{bare}}),
\label{eq:lemma}
\end{equation}
with equality if and only if $\rho_B$ is diagonal in the
logical basis.
\end{lemma}

\begin{proof}
%\textit{Proof.}
Writing $\rho_B=\bigl(\begin{smallmatrix}F_{\mathrm{bare}}&c\\
c^*&1-F_{\mathrm{bare}}\end{smallmatrix}\bigr)$ 
in the logical basis, the entropy $S(\rho_B)$
achieves its maximum value $H_2(F_{\mathrm{bare}})$ at $|c|=0$
(diagonal state) and decreases to zero as $|c|$ increases
to its maximum value $\sqrt{F_{\mathrm{bare}}(1-F_{\mathrm {bare}})}$
(pure state).
Hence $S(\rho_B)\le H_2(F_{\mathrm{bare}})$ for all $|c|$.
%\hfill$\square$
\end{proof}

\begin{theorem}[Theorem (No-go).]
\label{thm:nogo_SM}
%\textbf{\textit{Theorem (No-go).}}---
Let $|\Psi\rangle_{BE}$ be a pure state of a qubit block $B$ and an arbitrary environment $E$.
If
\begin{equation}
F_{\mathrm{bare}}
>
H_2^{-1}\!\bigl[(1-\delta)\ln 2\bigr],
\label{eq:nogo}
\end{equation}
no environment fragment can satisfy the Darwinism criterion, regardless of the Hamiltonian or environment structure. Consequently, Darwinistic redundancy cannot exist $(R_\delta =0)$. 
\end{theorem}

\begin{proof}
%\textit{Proof.}
Suppose, for contradiction, that Darwinistic redundancy
exists. Then some environment fragment satisfies
$\chi(F)\ge(1-\delta)\ln2$.
Using
\[
\chi(F)\le\chi(E)\le S(\rho_E)=S(\rho_B)
\]
together with Lemma~\ref{eq:lemma}, we obtain
\[
H_2(F_{\mathrm{bare}})
\ge
S(\rho_B)
=
S(\rho_E)
\ge
\chi(E)
\ge
\chi(F)
\ge
(1-\delta)\ln2.
\]
Since $H_2(p)$ is monotonically decreasing on
$p\in[1/2,1]$, this implies
\[
F_{\mathrm{bare}}
\le
H_2^{-1}\!\left[(1-\delta)\ln2\right],
\]
which contradicts Eq.~\eqref{eq:nogo}. Therefore no
environment fragment can satisfy the Darwinism criterion,
and hence Darwinistic redundancy cannot exist.
%\hfill$\square$
\end{proof}

\textbf{\textit{Corollary.}}---For the imperfect-recovery model,
$F_{\mathrm{bare}}=(F_L^{(N)}-\eta)/(1-\eta)$, so
Eq.~\eqref{eq:nogo} becomes
\begin{equation}
F_L^{(N)}
>
\eta+(1-\eta)\,H_2^{-1}\!\bigl[(1-\delta)\ln2\bigr]
\implies
R_\delta=0.
\label{eq:nogo_FL}
\end{equation}
For $\delta=0.10$ and $\eta=0.60$ this threshold is
$F_L^{(N)}>0.874$.

\textit{Remark.}---The solvable model saturates every
step of the proof chain.
The reduced block state is diagonal in the logical basis,
giving $S(\rho_B)=H_2(F_{\mathrm{bare}})$,
while the conditional environment states remain pure,
yielding $\chi(E)=S(\rho_E)$.
Together, these relations collapse the proof chain into the equality
$$
\chi(E)
=
S(\rho_E)
=
S(\rho_B)
=
H_2(F_{\mathrm{bare}}).
$$
Consequently, every bit of logical entropy generated in the
block is converted into accessible classical information in
the environment. The solvable model therefore achieves the
largest Darwinistic redundancy compatible with a given
logical fidelity, while the no-go theorem shows that no
other dynamics can exceed this information-theoretic limit.

% ================================
% PRL Section IV: Numerical Results
%==================================
\begin{figure}[htb]
\includegraphics[width=\columnwidth]{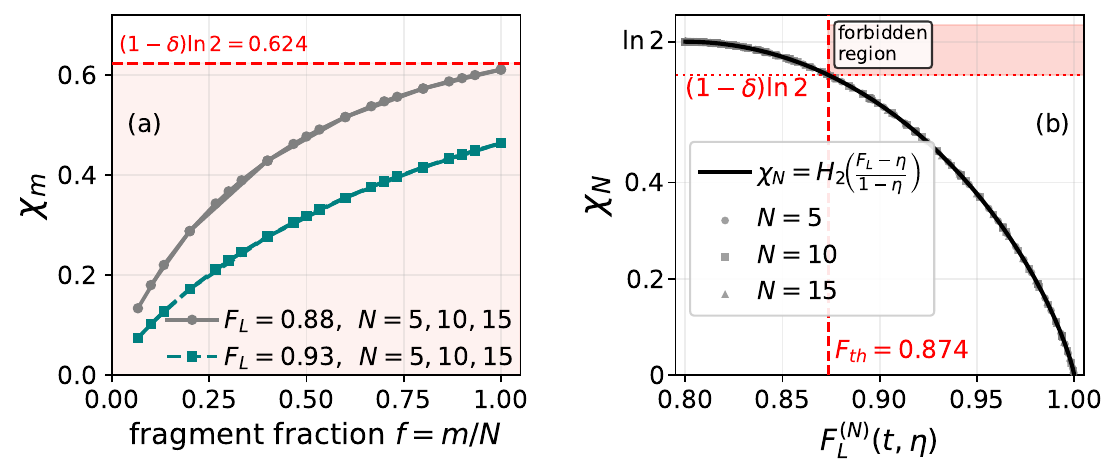}
\caption{
Numerical confirmation of the no-go theorem
, $\delta=0.10$, $\eta=0.60$).
\textit{(a)}~Fragment Holevo information $\chi_m$ as a function
of the fragment fraction $f=m/N$, computed for $N=5,10,15$
at two representative logical fidelities $F_L=0.88$
and $F_L=0.93$, both lying above the no-go threshold
$F_{\mathrm{th}}=0.874$.
Every curve remains strictly below the Darwinism threshold
$(1-\delta)\ln 2 = 0.624$~nat (red dashed line) for all
fragment sizes $m=1,\ldots,N$, demonstrating that no environment
fragment---however large---carries sufficient classical information
to establish Darwinistic redundancy when $F_L > F_{\mathrm{th}}$.
\textit{(b)}~Full-environment Holevo information $\chi_N$ as a
function of logical fidelity $F_L^{(N)}$, for $N=5,10,15$
at $g_X=0$ (closed-form, colored by $N$) and for $N=10$
at $g_X/g_Z=0.01,\,0.05,\,0.10$ (exact diagonalization,
distinct markers).
The black curve is the exact analytical result
$\chi_N = H_2\!\bigl((F_L^{(N)}-\eta)/(1-\eta)\bigr)$.
The salmon-shaded rectangle defines the forbidden region:
$F_L > F_{\mathrm{th}}$ and $\chi_N \ge (1-\delta)\ln 2$.
Of the data points spanning three environment sizes and
four coupling ratios ($g_X/g_Z=0$--$0.10$),
\emph{not a single point enters the forbidden rectangle},
confirming that the no-go bound is tight at $F_L=F_{\mathrm{th}}$
and is never violated.
}
\label{fig:nogo}
\end{figure}
% ============================================================
%Section IV: Numerical Results 
% ============================================================
\textbf{\textit{Numerical Validation.}}---We validate the
analytical predictions and test their robustness 
beyond the analytically solvable limit by exact diagonalization of the full Hamiltonian
$\hat{H}=g_Z\hat{Z}_b\!\otimes\!\hat{S}_Z
+g_X\hat{X}_b\!\otimes\!\hat{S}_X$
for several environment sizes and coupling ratios $g_X/g_Z$.

Figure~\ref{fig:competition}a illustrates the Darwinism--QEC
competition for $N=10$ at $g_X=0$. As the logical fidelity $F_L^{(N)}$ decreases toward its limiting value, the Darwinistic redundancy $R_\delta$ grows monotonically, in quantitative agreement with Eqs.~\eqref{eq:FL} and~\eqref{eq:R}.

Figure~\ref{fig:competition}b confirms the logarithmic critical
scaling of Eq.~\eqref{eq:log_sensitivity}: plotted as $R_\delta$ versus $F_L^{(N)}$ collapse onto the single analytical curve for different $N$, verifying that $N$ and $t$ enter only through
the fidelity gap $\varepsilon=F_L^{(N)}-F_c$ and demonstrating that the tradeoff is independent of the microscopic dynamical parameters.
The redundancy diverges as $F_L^{(N)}\to F_c^+$ and vanishes
above the no-go threshold $F_L^{(N)}>0.874$, both in excellent
agreement with the analytical predictions.

Figures~\ref{fig:competition}c--d demonstrate the robustness
of the analytical tradeoff against the non-commuting
perturbation $g_X\hat{X}_b\otimes\hat{S}_X$.
For $g_X/g_Z=0.01$, deviations remain below $0.2\%$;
for $g_X/g_Z=0.05$ ($0.10$) they reach $4\%$ ($12\%$),
growing as $g_X$ approaches the solvable interaction
strength.

Figure~\ref{fig:nogo} provides numerical confirmation of the
model-independent no-go theorem.
Panel~(a) shows that, for logical fidelities above the
predicted threshold, the Holevo information of every
environment fragment remains below the Darwinism criterion,
independent of fragment size.
Panel~(b) demonstrates that numerical data obtained for
different environment sizes and coupling ratios never enter
the forbidden region predicted by
Theorem \ref{thm:nogo_SM}, confirming the general bound beyond
the exactly solvable limit. 

% ============================================================
% Section V: Discussion
% ============================================================
\textbf{\textit{Discussion.}}---
 The present work places the relationship between quantum error correction and Quantum Darwinism on a quantitative footing by establishing a direct information-theoretic bridge between the two.
Within an exactly solvable block--environment model, we derive a parameter-free tradeoff between the post-recovery logical fidelity and Darwinistic redundancy. Because both quantities are governed by the same decoherence parameter, the logical fidelity uniquely determines the Holevo information of environment fragments—and hence the Darwinistic redundancy—establishing logical recovery and environment-fragment readout as equivalent operational probes of the same underlying decoherence process. Because the solvable model saturates the underlying information-theoretic bounds, every loss of logical information is exactly balanced by a corresponding gain in classically accessible environmental records.

Beyond the solvable dynamics, the no-go theorem establishes a model-independent information-theoretic constraint on the coexistence of quantum error correction and Quantum Darwinism: the existence of Darwinistic redundancy requires the bare logical fidelity to lie below a critical threshold, independent of the microscopic Hamiltonian or environment structure. For the imperfect-recovery
channel considered here, this general bound translates directly
into a threshold for the post-recovery logical fidelity, while
the solvable model saturates the bound, demonstrating that it is
an extremal realization of the Darwinism--QEC competition.
Natural extensions include correlated or initially entangled
environments, spatially local interactions, general recovery
channels, and higher-dimensional logical encodings. 
More generally, these results suggest that decoherence, quantum error correction, and the emergence of objective classical information are governed by common information-theoretic principles rather than by the details of a particular microscopic model, providing a unified framework for understanding the quantum-to-classical transition in error-corrected quantum systems.

%==================================================
\section{Acknowledgement}

This research is supported by the Ministry of Education, Singapore, under its Academic Research Fund Programme (T2EP50222-0038, RG182/25 and RG154/24). 

\bibliography{qec_bibfile}

\begin{thebibliography}{25}%
\makeatletter
\providecommand \@ifxundefined [1]{%
 \@ifx{#1\undefined}
}%
\providecommand \@ifnum [1]{%
 \ifnum #1\expandafter \@firstoftwo
 \else \expandafter \@secondoftwo
 \fi
}%
\providecommand \@ifx [1]{%
 \ifx #1\expandafter \@firstoftwo
 \else \expandafter \@secondoftwo
 \fi
}%
\providecommand \natexlab [1]{#1}%
\providecommand \enquote  [1]{``#1''}%
\providecommand \bibnamefont  [1]{#1}%
\providecommand \bibfnamefont [1]{#1}%
\providecommand \citenamefont [1]{#1}%
\providecommand \href@noop [0]{\@secondoftwo}%
\providecommand \href [0]{\begingroup \@sanitize@url \@href}%
\providecommand \@href[1]{\@@startlink{#1}\@@href}%
\providecommand \@@href[1]{\endgroup#1\@@endlink}%
\providecommand \@sanitize@url [0]{\catcode `\\12\catcode `\$12\catcode `\&12\catcode `\#12\catcode `\^12\catcode `\_12\catcode `\%12\relax}%
\providecommand \@@startlink[1]{}%
\providecommand \@@endlink[0]{}%
\providecommand \url  [0]{\begingroup\@sanitize@url \@url }%
\providecommand \@url [1]{\endgroup\@href {#1}{\urlprefix }}%
\providecommand \urlprefix  [0]{URL }%
\providecommand \Eprint [0]{\href }%
\providecommand \doibase [0]{https://doi.org/}%
\providecommand \selectlanguage [0]{\@gobble}%
\providecommand \bibinfo  [0]{\@secondoftwo}%
\providecommand \bibfield  [0]{\@secondoftwo}%
\providecommand \translation [1]{[#1]}%
\providecommand \BibitemOpen [0]{}%
\providecommand \bibitemStop [0]{}%
\providecommand \bibitemNoStop [0]{.\EOS\space}%
\providecommand \EOS [0]{\spacefactor3000\relax}%
\providecommand \BibitemShut  [1]{\csname bibitem#1\endcsname}%
\let\auto@bib@innerbib\@empty
%</preamble>
\bibitem [{\citenamefont {Shor}(1995)}]{Shor_PRA_1995}%
  \BibitemOpen
  \bibfield  {author} {\bibinfo {author} {\bibfnamefont {P.~W.}\ \bibnamefont {Shor}},\ }\bibfield  {title} {\bibinfo {title} {Scheme for reducing decoherence in quantum computer memory},\ }\href {https://doi.org/10.1103/PhysRevA.52.R2493} {\bibfield  {journal} {\bibinfo  {journal} {Phys. Rev. A}\ }\textbf {\bibinfo {volume} {52}},\ \bibinfo {pages} {R2493(R)} (\bibinfo {year} {1995})}\BibitemShut {NoStop}%
\bibitem [{\citenamefont {Steane}(1996)}]{Steane_PRL_1996}%
  \BibitemOpen
  \bibfield  {author} {\bibinfo {author} {\bibfnamefont {A.~M.}\ \bibnamefont {Steane}},\ }\bibfield  {title} {\bibinfo {title} {Error correcting codes in quantum theory},\ }\href {https://doi.org/10.1103/PhysRevLett.77.793} {\bibfield  {journal} {\bibinfo  {journal} {Phys. Rev. Lett.}\ }\textbf {\bibinfo {volume} {77}},\ \bibinfo {pages} {793} (\bibinfo {year} {1996})}\BibitemShut {NoStop}%
\bibitem [{\citenamefont {Calderbank}\ and\ \citenamefont {Shor}(1996)}]{Shor_PRA_1996}%
  \BibitemOpen
  \bibfield  {author} {\bibinfo {author} {\bibfnamefont {A.~R.}\ \bibnamefont {Calderbank}}\ and\ \bibinfo {author} {\bibfnamefont {P.~W.}\ \bibnamefont {Shor}},\ }\bibfield  {title} {\bibinfo {title} {Good quantum error-correcting codes exist},\ }\href {https://doi.org/10.1103/PhysRevA.54.1098} {\bibfield  {journal} {\bibinfo  {journal} {Phys. Rev. A}\ }\textbf {\bibinfo {volume} {54}},\ \bibinfo {pages} {1098} (\bibinfo {year} {1996})}\BibitemShut {NoStop}%
\bibitem [{\citenamefont {Laflamme}\ \emph {et~al.}(1996)\citenamefont {Laflamme}, \citenamefont {Miquel}, \citenamefont {Paz},\ and\ \citenamefont {Zurek}}]{Zurek_PRL_1996}%
  \BibitemOpen
  \bibfield  {author} {\bibinfo {author} {\bibfnamefont {R.}~\bibnamefont {Laflamme}}, \bibinfo {author} {\bibfnamefont {C.}~\bibnamefont {Miquel}}, \bibinfo {author} {\bibfnamefont {J.~P.}\ \bibnamefont {Paz}},\ and\ \bibinfo {author} {\bibfnamefont {W.~H.}\ \bibnamefont {Zurek}},\ }\bibfield  {title} {\bibinfo {title} {Perfect quantum error correcting code},\ }\href {https://doi.org/10.1103/PhysRevLett.77.198} {\bibfield  {journal} {\bibinfo  {journal} {Phys. Rev. Lett.}\ }\textbf {\bibinfo {volume} {77}},\ \bibinfo {pages} {198} (\bibinfo {year} {1996})}\BibitemShut {NoStop}%
\bibitem [{\citenamefont {Knill}\ and\ \citenamefont {Laflamme}(1997)}]{Knill_PRA_1997}%
  \BibitemOpen
  \bibfield  {author} {\bibinfo {author} {\bibfnamefont {E.}~\bibnamefont {Knill}}\ and\ \bibinfo {author} {\bibfnamefont {R.}~\bibnamefont {Laflamme}},\ }\bibfield  {title} {\bibinfo {title} {Theory of quantum error-correcting codes},\ }\href {https://doi.org/10.1103/PhysRevA.55.900} {\bibfield  {journal} {\bibinfo  {journal} {Phys. Rev. A}\ }\textbf {\bibinfo {volume} {55}},\ \bibinfo {pages} {900} (\bibinfo {year} {1997})}\BibitemShut {NoStop}%
\bibitem [{\citenamefont {Gottesman}(1997)}]{Gottesman_Thesis_1997}%
  \BibitemOpen
  \bibfield  {author} {\bibinfo {author} {\bibfnamefont {D.}~\bibnamefont {Gottesman}},\ }\href {https://arxiv.org/abs/quant-ph/9705052} {\bibinfo {title} {Stabilizer codes and quantum error correction}} (\bibinfo {year} {1997}),\ \Eprint {https://arxiv.org/abs/quant-ph/9705052} {arXiv:quant-ph/9705052 [quant-ph]} \BibitemShut {NoStop}%
\bibitem [{\citenamefont {Calderbank}\ \emph {et~al.}(1998)\citenamefont {Calderbank}, \citenamefont {Rains}, \citenamefont {Shor},\ and\ \citenamefont {Sloane}}]{Shor_IEEE_1998}%
  \BibitemOpen
  \bibfield  {author} {\bibinfo {author} {\bibfnamefont {A.}~\bibnamefont {Calderbank}}, \bibinfo {author} {\bibfnamefont {E.}~\bibnamefont {Rains}}, \bibinfo {author} {\bibfnamefont {P.}~\bibnamefont {Shor}},\ and\ \bibinfo {author} {\bibfnamefont {N.}~\bibnamefont {Sloane}},\ }\bibfield  {title} {\bibinfo {title} {Quantum error correction via codes over gf(4)},\ }\href {https://doi.org/10.1109/18.681315} {\bibfield  {journal} {\bibinfo  {journal} {IEEE Transactions on Information Theory}\ }\textbf {\bibinfo {volume} {44}},\ \bibinfo {pages} {1369} (\bibinfo {year} {1998})}\BibitemShut {NoStop}%
\bibitem [{\citenamefont {Zurek}(2003)}]{Zurek_RMP_2003}%
  \BibitemOpen
  \bibfield  {author} {\bibinfo {author} {\bibfnamefont {W.~H.}\ \bibnamefont {Zurek}},\ }\bibfield  {title} {\bibinfo {title} {Decoherence, einselection, and the quantum origins of the classical},\ }\href {https://doi.org/10.1103/RevModPhys.75.715} {\bibfield  {journal} {\bibinfo  {journal} {Rev. Mod. Phys.}\ }\textbf {\bibinfo {volume} {75}},\ \bibinfo {pages} {715} (\bibinfo {year} {2003})}\BibitemShut {NoStop}%
\bibitem [{\citenamefont {Ollivier}\ \emph {et~al.}(2005)\citenamefont {Ollivier}, \citenamefont {Poulin},\ and\ \citenamefont {Zurek}}]{Zurek_PRA_2005}%
  \BibitemOpen
  \bibfield  {author} {\bibinfo {author} {\bibfnamefont {H.}~\bibnamefont {Ollivier}}, \bibinfo {author} {\bibfnamefont {D.}~\bibnamefont {Poulin}},\ and\ \bibinfo {author} {\bibfnamefont {W.~H.}\ \bibnamefont {Zurek}},\ }\bibfield  {title} {\bibinfo {title} {Environment as a witness: Selective proliferation of information and emergence of objectivity in a quantum universe},\ }\href {https://doi.org/10.1103/PhysRevA.72.042113} {\bibfield  {journal} {\bibinfo  {journal} {Phys. Rev. A}\ }\textbf {\bibinfo {volume} {72}},\ \bibinfo {pages} {042113} (\bibinfo {year} {2005})}\BibitemShut {NoStop}%
\bibitem [{\citenamefont {Zurek}(2009)}]{Zurek_NP_2009}%
  \BibitemOpen
  \bibfield  {author} {\bibinfo {author} {\bibfnamefont {W.~H.}\ \bibnamefont {Zurek}},\ }\bibfield  {title} {\bibinfo {title} {Quantum darwinism},\ }\href {https://doi.org/10.1038/nphys1202} {\bibfield  {journal} {\bibinfo  {journal} {Nature Physics}\ }\textbf {\bibinfo {volume} {5}},\ \bibinfo {pages} {181} (\bibinfo {year} {2009})}\BibitemShut {NoStop}%
\bibitem [{\citenamefont {Zwolak}\ \emph {et~al.}(2009)\citenamefont {Zwolak}, \citenamefont {Quan},\ and\ \citenamefont {Zurek}}]{Zurek_PRL_2009}%
  \BibitemOpen
  \bibfield  {author} {\bibinfo {author} {\bibfnamefont {M.}~\bibnamefont {Zwolak}}, \bibinfo {author} {\bibfnamefont {H.~T.}\ \bibnamefont {Quan}},\ and\ \bibinfo {author} {\bibfnamefont {W.~H.}\ \bibnamefont {Zurek}},\ }\bibfield  {title} {\bibinfo {title} {Quantum darwinism in a mixed environment},\ }\href {https://doi.org/10.1103/PhysRevLett.103.110402} {\bibfield  {journal} {\bibinfo  {journal} {Phys. Rev. Lett.}\ }\textbf {\bibinfo {volume} {103}},\ \bibinfo {pages} {110402} (\bibinfo {year} {2009})}\BibitemShut {NoStop}%
\bibitem [{\citenamefont {Unden}\ \emph {et~al.}(2019)\citenamefont {Unden}, \citenamefont {Louzon}, \citenamefont {Zwolak}, \citenamefont {Zurek},\ and\ \citenamefont {Jelezko}}]{Jelezko_PRL_2019}%
  \BibitemOpen
  \bibfield  {author} {\bibinfo {author} {\bibfnamefont {T.~K.}\ \bibnamefont {Unden}}, \bibinfo {author} {\bibfnamefont {D.}~\bibnamefont {Louzon}}, \bibinfo {author} {\bibfnamefont {M.}~\bibnamefont {Zwolak}}, \bibinfo {author} {\bibfnamefont {W.~H.}\ \bibnamefont {Zurek}},\ and\ \bibinfo {author} {\bibfnamefont {F.}~\bibnamefont {Jelezko}},\ }\bibfield  {title} {\bibinfo {title} {Revealing the emergence of classicality using nitrogen-vacancy centers},\ }\href {https://doi.org/10.1103/PhysRevLett.123.140402} {\bibfield  {journal} {\bibinfo  {journal} {Phys. Rev. Lett.}\ }\textbf {\bibinfo {volume} {123}},\ \bibinfo {pages} {140402} (\bibinfo {year} {2019})}\BibitemShut {NoStop}%
\bibitem [{\citenamefont {Ciampini}\ \emph {et~al.}(2018)\citenamefont {Ciampini}, \citenamefont {Pinna}, \citenamefont {Mataloni},\ and\ \citenamefont {Paternostro}}]{Paternostro_PRA_2018}%
  \BibitemOpen
  \bibfield  {author} {\bibinfo {author} {\bibfnamefont {M.~A.}\ \bibnamefont {Ciampini}}, \bibinfo {author} {\bibfnamefont {G.}~\bibnamefont {Pinna}}, \bibinfo {author} {\bibfnamefont {P.}~\bibnamefont {Mataloni}},\ and\ \bibinfo {author} {\bibfnamefont {M.}~\bibnamefont {Paternostro}},\ }\bibfield  {title} {\bibinfo {title} {Experimental signature of quantum darwinism in photonic cluster states},\ }\href {https://doi.org/10.1103/PhysRevA.98.020101} {\bibfield  {journal} {\bibinfo  {journal} {Phys. Rev. A}\ }\textbf {\bibinfo {volume} {98}},\ \bibinfo {pages} {020101(R)} (\bibinfo {year} {2018})}\BibitemShut {NoStop}%
\bibitem [{\citenamefont {Chen}\ \emph {et~al.}(2019)\citenamefont {Chen}, \citenamefont {Zhong}, \citenamefont {Li}, \citenamefont {Wu}, \citenamefont {Wang}, \citenamefont {Li}, \citenamefont {Liu}, \citenamefont {Lu},\ and\ \citenamefont {Pan}}]{Pan_Sci-Bulletin_2019}%
  \BibitemOpen
  \bibfield  {author} {\bibinfo {author} {\bibfnamefont {M.-C.}\ \bibnamefont {Chen}}, \bibinfo {author} {\bibfnamefont {H.-S.}\ \bibnamefont {Zhong}}, \bibinfo {author} {\bibfnamefont {Y.}~\bibnamefont {Li}}, \bibinfo {author} {\bibfnamefont {D.}~\bibnamefont {Wu}}, \bibinfo {author} {\bibfnamefont {X.-L.}\ \bibnamefont {Wang}}, \bibinfo {author} {\bibfnamefont {L.}~\bibnamefont {Li}}, \bibinfo {author} {\bibfnamefont {N.-L.}\ \bibnamefont {Liu}}, \bibinfo {author} {\bibfnamefont {C.-Y.}\ \bibnamefont {Lu}},\ and\ \bibinfo {author} {\bibfnamefont {J.-W.}\ \bibnamefont {Pan}},\ }\bibfield  {title} {\bibinfo {title} {Emergence of classical objectivity of quantum darwinism in a photonic quantum simulator},\ }\href {https://doi.org/https://doi.org/10.1016/j.scib.2019.03.032} {\bibfield  {journal} {\bibinfo  {journal} {Science Bulletin}\ }\textbf {\bibinfo {volume} {64}},\ \bibinfo {pages} {580} (\bibinfo {year} {2019})}\BibitemShut {NoStop}%
\bibitem [{\citenamefont {Girard}\ \emph {et~al.}(2026)\citenamefont {Girard}, \citenamefont {Cheng},\ and\ \citenamefont {Cao}}]{Girard_arXiv_2026}%
  \BibitemOpen
  \bibfield  {author} {\bibinfo {author} {\bibfnamefont {M.}~\bibnamefont {Girard}}, \bibinfo {author} {\bibfnamefont {G.}~\bibnamefont {Cheng}},\ and\ \bibinfo {author} {\bibfnamefont {C.}~\bibnamefont {Cao}},\ }\href {https://arxiv.org/abs/2606.06588} {\bibinfo {title} {Demystifying objectivity with operator algebra quantum error correction}} (\bibinfo {year} {2026}),\ \Eprint {https://arxiv.org/abs/2606.06588} {arXiv:2606.06588 [quant-ph]} \BibitemShut {NoStop}%
\bibitem [{\citenamefont {Cory}\ \emph {et~al.}(1998)\citenamefont {Cory}, \citenamefont {Price}, \citenamefont {Maas}, \citenamefont {Knill}, \citenamefont {Laflamme}, \citenamefont {Zurek}, \citenamefont {Havel},\ and\ \citenamefont {Somaroo}}]{Somaroo_PRL_1998}%
  \BibitemOpen
  \bibfield  {author} {\bibinfo {author} {\bibfnamefont {D.~G.}\ \bibnamefont {Cory}}, \bibinfo {author} {\bibfnamefont {M.~D.}\ \bibnamefont {Price}}, \bibinfo {author} {\bibfnamefont {W.}~\bibnamefont {Maas}}, \bibinfo {author} {\bibfnamefont {E.}~\bibnamefont {Knill}}, \bibinfo {author} {\bibfnamefont {R.}~\bibnamefont {Laflamme}}, \bibinfo {author} {\bibfnamefont {W.~H.}\ \bibnamefont {Zurek}}, \bibinfo {author} {\bibfnamefont {T.~F.}\ \bibnamefont {Havel}},\ and\ \bibinfo {author} {\bibfnamefont {S.~S.}\ \bibnamefont {Somaroo}},\ }\bibfield  {title} {\bibinfo {title} {Experimental quantum error correction},\ }\href {https://doi.org/10.1103/PhysRevLett.81.2152} {\bibfield  {journal} {\bibinfo  {journal} {Phys. Rev. Lett.}\ }\textbf {\bibinfo {volume} {81}},\ \bibinfo {pages} {2152} (\bibinfo {year} {1998})}\BibitemShut {NoStop}%
\bibitem [{\citenamefont {Yao}\ \emph {et~al.}(2012)\citenamefont {Yao}, \citenamefont {Wang}, \citenamefont {Chen}, \citenamefont {Gao}, \citenamefont {Fowler}, \citenamefont {Raussendorf}, \citenamefont {Chen}, \citenamefont {Liu}, \citenamefont {Lu}, \citenamefont {Deng}, \citenamefont {Chen},\ and\ \citenamefont {Pan}}]{Pan_Nature_2012}%
  \BibitemOpen
  \bibfield  {author} {\bibinfo {author} {\bibfnamefont {X.-C.}\ \bibnamefont {Yao}}, \bibinfo {author} {\bibfnamefont {T.-X.}\ \bibnamefont {Wang}}, \bibinfo {author} {\bibfnamefont {H.-Z.}\ \bibnamefont {Chen}}, \bibinfo {author} {\bibfnamefont {W.-B.}\ \bibnamefont {Gao}}, \bibinfo {author} {\bibfnamefont {A.~G.}\ \bibnamefont {Fowler}}, \bibinfo {author} {\bibfnamefont {R.}~\bibnamefont {Raussendorf}}, \bibinfo {author} {\bibfnamefont {Z.-B.}\ \bibnamefont {Chen}}, \bibinfo {author} {\bibfnamefont {N.-L.}\ \bibnamefont {Liu}}, \bibinfo {author} {\bibfnamefont {C.-Y.}\ \bibnamefont {Lu}}, \bibinfo {author} {\bibfnamefont {Y.-J.}\ \bibnamefont {Deng}}, \bibinfo {author} {\bibfnamefont {Y.-A.}\ \bibnamefont {Chen}},\ and\ \bibinfo {author} {\bibfnamefont {J.-W.}\ \bibnamefont {Pan}},\ }\bibfield  {title} {\bibinfo {title} {Experimental demonstration of topological error correction},\ }\href {https://doi.org/10.1038/nature10770} {\bibfield  {journal} {\bibinfo  {journal} {Nature}\ }\textbf {\bibinfo {volume}
  {482}},\ \bibinfo {pages} {489} (\bibinfo {year} {2012})}\BibitemShut {NoStop}%
\bibitem [{\citenamefont {C{\'o}rcoles}\ \emph {et~al.}(2015)\citenamefont {C{\'o}rcoles}, \citenamefont {Magesan}, \citenamefont {Srinivasan}, \citenamefont {Cross}, \citenamefont {Steffen}, \citenamefont {Gambetta},\ and\ \citenamefont {Chow}}]{Chow_Nature-Comm_2015}%
  \BibitemOpen
  \bibfield  {author} {\bibinfo {author} {\bibfnamefont {A.~D.}\ \bibnamefont {C{\'o}rcoles}}, \bibinfo {author} {\bibfnamefont {E.}~\bibnamefont {Magesan}}, \bibinfo {author} {\bibfnamefont {S.~J.}\ \bibnamefont {Srinivasan}}, \bibinfo {author} {\bibfnamefont {A.~W.}\ \bibnamefont {Cross}}, \bibinfo {author} {\bibfnamefont {M.}~\bibnamefont {Steffen}}, \bibinfo {author} {\bibfnamefont {J.~M.}\ \bibnamefont {Gambetta}},\ and\ \bibinfo {author} {\bibfnamefont {J.~M.}\ \bibnamefont {Chow}},\ }\bibfield  {title} {\bibinfo {title} {Demonstration of a quantum error detection code using a square lattice of four superconducting qubits},\ }\href {https://doi.org/10.1038/ncomms7979} {\bibfield  {journal} {\bibinfo  {journal} {Nature Communications}\ }\textbf {\bibinfo {volume} {6}},\ \bibinfo {pages} {6979} (\bibinfo {year} {2015})}\BibitemShut {NoStop}%
\bibitem [{\citenamefont {Ofek}\ \emph {et~al.}(2016)\citenamefont {Ofek} \emph {et~al.}}]{Schoelkopf_Nature_2016}%
  \BibitemOpen
  \bibfield  {author} {\bibinfo {author} {\bibfnamefont {N.}~\bibnamefont {Ofek}} \emph {et~al.},\ }\bibfield  {title} {\bibinfo {title} {Extending the lifetime of a quantum bit with error correction in superconducting circuits},\ }\href {https://doi.org/10.1038/nature18949} {\bibfield  {journal} {\bibinfo  {journal} {Nature}\ }\textbf {\bibinfo {volume} {536}},\ \bibinfo {pages} {441} (\bibinfo {year} {2016})}\BibitemShut {NoStop}%
\bibitem [{\citenamefont {Acharya}\ \emph {et~al.}(2023)\citenamefont {Acharya} \emph {et~al.}}]{Google_Nature_2023}%
  \BibitemOpen
  \bibfield  {author} {\bibinfo {author} {\bibfnamefont {R.}~\bibnamefont {Acharya}} \emph {et~al.},\ }\bibfield  {title} {\bibinfo {title} {Suppressing quantum errors by scaling a surface code logical qubit},\ }\href {https://doi.org/10.1038/s41586-022-05434-1} {\bibfield  {journal} {\bibinfo  {journal} {Nature}\ }\textbf {\bibinfo {volume} {614}},\ \bibinfo {pages} {676} (\bibinfo {year} {2023})}\BibitemShut {NoStop}%
\bibitem [{\citenamefont {DeGrace}\ \emph {et~al.}(2022)\citenamefont {DeGrace} \emph {et~al.}}]{IBM_Nature_2022}%
  \BibitemOpen
  \bibfield  {author} {\bibinfo {author} {\bibfnamefont {M.~M.}\ \bibnamefont {DeGrace}} \emph {et~al.},\ }\bibfield  {title} {\bibinfo {title} {Defining the risk of sars-cov-2 variants on immune protection},\ }\href {https://doi.org/10.1038/s41586-022-04690-5} {\bibfield  {journal} {\bibinfo  {journal} {Nature}\ }\textbf {\bibinfo {volume} {605}},\ \bibinfo {pages} {640} (\bibinfo {year} {2022})}\BibitemShut {NoStop}%
\bibitem [{\citenamefont {Ryan-Anderson}\ \emph {et~al.}(2021)\citenamefont {Ryan-Anderson} \emph {et~al.}}]{Anderson_PRX_2021}%
  \BibitemOpen
  \bibfield  {author} {\bibinfo {author} {\bibfnamefont {C.}~\bibnamefont {Ryan-Anderson}} \emph {et~al.},\ }\bibfield  {title} {\bibinfo {title} {Realization of real-time fault-tolerant quantum error correction},\ }\href {https://doi.org/10.1103/PhysRevX.11.041058} {\bibfield  {journal} {\bibinfo  {journal} {Phys. Rev. X}\ }\textbf {\bibinfo {volume} {11}},\ \bibinfo {pages} {041058} (\bibinfo {year} {2021})}\BibitemShut {NoStop}%
\bibitem [{\citenamefont {Zurek}(1982)}]{Zurek_PRD_1982}%
  \BibitemOpen
  \bibfield  {author} {\bibinfo {author} {\bibfnamefont {W.~H.}\ \bibnamefont {Zurek}},\ }\bibfield  {title} {\bibinfo {title} {Environment-induced superselection rules},\ }\href {https://doi.org/10.1103/PhysRevD.26.1862} {\bibfield  {journal} {\bibinfo  {journal} {Phys. Rev. D}\ }\textbf {\bibinfo {volume} {26}},\ \bibinfo {pages} {1862} (\bibinfo {year} {1982})}\BibitemShut {NoStop}%
\bibitem [{\citenamefont {Le}\ and\ \citenamefont {Olaya-Castro}(2019)}]{Castro_PRL_2019}%
  \BibitemOpen
  \bibfield  {author} {\bibinfo {author} {\bibfnamefont {T.~P.}\ \bibnamefont {Le}}\ and\ \bibinfo {author} {\bibfnamefont {A.}~\bibnamefont {Olaya-Castro}},\ }\bibfield  {title} {\bibinfo {title} {Strong quantum darwinism and strong independence are equivalent to spectrum broadcast structure},\ }\href {https://doi.org/10.1103/PhysRevLett.122.010403} {\bibfield  {journal} {\bibinfo  {journal} {Phys. Rev. Lett.}\ }\textbf {\bibinfo {volume} {122}},\ \bibinfo {pages} {010403} (\bibinfo {year} {2019})}\BibitemShut {NoStop}%
\bibitem [{\citenamefont {Korbicz}(2021)}]{Korbicz_Quantum_2021}%
  \BibitemOpen
  \bibfield  {author} {\bibinfo {author} {\bibfnamefont {J.~K.}\ \bibnamefont {Korbicz}},\ }\bibfield  {title} {\bibinfo {title} {Roads to objectivity: {Q}uantum {D}arwinism, {S}pectrum {B}roadcast {S}tructures, and {S}trong quantum {D}arwinism – a review},\ }\href {https://doi.org/10.22331/q-2021-11-08-571} {\bibfield  {journal} {\bibinfo  {journal} {{Quantum}}\ }\textbf {\bibinfo {volume} {5}},\ \bibinfo {pages} {571} (\bibinfo {year} {2021})}\BibitemShut {NoStop}%
\end{thebibliography}%

\appendix
\begin{widetext}

% ---- Subsection A ----
\section{Collective Hamiltonian}
\label{sec:collective_H}

The main-text model involves a logical block $b$ coupled to
$N$ independent environment qubits through the Hamiltonian
\begin{equation}
\hat{H}_N
=
g_Z\,\hat{Z}_b\otimes\hat{S}_Z
+
g_X\,\hat{X}_b\otimes\hat{S}_X,
\label{eq:HN_SM}
\end{equation}
where $\hat{S}_{Z(X)}=\sum_{k=1}^N Z_k(X_k)$ are collective
environment spin operators, and $\hat{Z}_b=Z^{\otimes 3}$,
$\hat{X}_b=X^{\otimes 3}$ act on the logical block.
The $g_Z$ interaction produces pure dephasing in the eigenbasis
of $\hat{Z}_b$, giving rise to logical decoherence in the
logical basis, while the $g_X$ interaction induces coherent
logical-state mixing that acts as an uncorrectable logical
error channel. The two terms therefore describe the dominant
dephasing interaction ($g_Z$) together with a non-commuting
perturbation ($g_X$), providing the minimal model used to
assess the robustness of the analytical results.

% ---- Subsection B ----
\section{Many-environment spectrum}
\label{sec:nontrivial}

For a single environment qubit ($N=1$), $\hat{H}_1^2$ is
proportional to the identity on the environment, giving the
three-level spectrum $\{0,\pm\sqrt{g_Z^2+g_X^2}\}$ that
underlies the exact closed-form solutions of the main text.
For $N\ge2$ this simplification no longer holds.
A direct calculation using $\hat{Z}_b^2=\hat{X}_b^2=\mathbb{I}_b$
and the Pauli relation $\hat{Z}_b\hat{X}_b+\hat{X}_b\hat{Z}_b=0$
gives
\begin{equation}
\hat{H}_N^2
=
\mathbb{I}_b\otimes
\bigl(g_Z^2\hat{S}_Z^2+g_X^2\hat{S}_X^2\bigr)
-
2g_Zg_X\,\hat{Y}_b\otimes\hat{S}_Y,
\label{eq:HN2}
\end{equation}
where we used the collective commutator
$[\hat{S}_X,\hat{S}_Z]=-[\hat{S}_Z,\hat{S}_X]=2i\hat{S}_Y$.
For $N\ge2$ the operator
$g_Z^2\hat{S}_Z^2+g_X^2\hat{S}_X^2$ is not proportional to
the environment identity; for example, with $N=2$,
\begin{equation}
g_Z^2\hat{S}_Z^2+g_X^2\hat{S}_X^2
=
(g_Z^2+g_X^2)\cdot 2\mathbb{I}
+
2g_Z^2 Z_1Z_2
+
2g_X^2 X_1X_2,
\end{equation}
which carries a nontrivial two-qubit spectrum.
Consequently, closed-form diagonalization is no longer
available in general, and exact diagonalization provides
the numerical approach used throughout this work.

% ---- Subsection C ----
\section{Weak-coupling approximation}
\label{sec:weakcoupling}

To obtain analytical progress beyond $g_X=0$ we decompose
\begin{equation}
\hat{H}_N = \hat{H}_0 + \hat{V},
\qquad
\hat{H}_0 = g_Z\hat{Z}_b\otimes\hat{S}_Z,
\qquad
\hat{V} = g_X\hat{X}_b\otimes\hat{S}_X,
\label{eq:perturbation}
\end{equation}
and treat $\hat{V}$ as a perturbation valid for
$g_X/g_Z\ll1$.
The unperturbed Hamiltonian $\hat{H}_0$ is exactly solvable
for arbitrary $N$ (main text, Sec.~II), providing a
systematic starting point for the perturbative analysis which forms the basis of the perturbative analysis presented
below.

% ============================================================
\section{Exact dynamics under $\hat{H}_0$}
\label{sec:H0}
% ============================================================
We derive the time-evolved state under the unperturbed
Hamiltonian
$\hat{H}_0=g_Z\hat{Z}_b\otimes\hat{S}_Z$.
The logical block is initialized in the logical codeword
$|\bar{z}_+\rangle_b=(|000\rangle+|111\rangle)/\sqrt2$,
corresponding to the encoded logical qubit, while the
environment is assumed to begin in an uncorrelated product
state with each qubit prepared in
$|+\rangle=(|0\rangle+|1\rangle)/\sqrt2$.
The initial state is therefore
\[
|\Psi(0)\rangle
=
|\bar{z}_+\rangle_b
\otimes
|+\rangle^{\otimes N}.
\]

Since $\hat{Z}_b^2 = \mathbb{I}_b$ and $[Z_i, Z_j] = 0$ for
all $i\neq j$, the operator powers satisfy
$\hat{H}_0^n = g_Z^n\hat{Z}_b^n\otimes\hat{S}_Z^n$.
Expanding the propagator as a Taylor series and separating
even and odd powers, using
$\hat{Z}_b^{2m}|\bar{z}_+\rangle = |\bar{z}_+\rangle$ and
$\hat{Z}_b^{2m+1}|\bar{z}_+\rangle = |\bar{z}_-\rangle$, gives
$$
|\Psi_{+}(t)\rangle
=
e^{-it\hat H_0}
|\Psi(0)\rangle.
$$
\begin{equation}
|\Psi_{+}(t)\rangle
=
|\bar{z}_+\rangle_b
\otimes
\sum_{m=0}^{\infty}
\frac{(-1)^m(g_Z t)^{2m}}{(2m)!}
\hat{S}_Z^{2m}
|{+}\rangle^{\otimes N}
-i\,|\bar{z}_-\rangle_b
\otimes
\sum_{m=0}^{\infty}
\frac{(-1)^m(g_Z t)^{2m+1}}{(2m+1)!}
\hat{S}_Z^{2m+1}
|{+}\rangle^{\otimes N}.
\end{equation}
The subscript $``+"$ indicates evolution from the initial logical state $|\bar{z}_+\rangle$.
Recognizing the operator cosine and sine series yields the
exact evolved state
\begin{equation}
|\Psi_{+}(t)\rangle
=
|\bar{z}_+\rangle_b\otimes
\cos(g_Z t\hat{S}_Z)|{+}\rangle^{\otimes N}
-
i\,|\bar{z}_-\rangle_b\otimes
\sin(g_Z t\hat{S}_Z)|{+}\rangle^{\otimes N},
\label{eq:psi_exact}
\end{equation}
and the corresponding exact propagator
\begin{equation}
\hat{U}_N(t)
=
\cos(g_Z t\hat{S}_Z)\,\mathbb{I}_b
-
i\hat{Z}_b\otimes\sin(g_Z t\hat{S}_Z).
\label{eq:UN}
\end{equation}
%Because $[Z_i, Z_j]=0$ for all pairs, the operator functions $\cos(g_Z t\hat{S}_Z)$ and $\sin(g_Z t\hat{S}_Z)$ factorize over the individual environment qubits:
Since $ \hat S_Z=\sum_{k=1}^N Z_k, $
and the operators $Z_k$ mutually commute,
the exponential factorizes as 
$
e^{-ig_Zt\hat S_Z}
=
\prod_{k=1}^N
e^{-ig_ZtZ_k},
$
from which the corresponding factorization of
$\cos(g_Zt\hat S_Z)$ and
$\sin(g_Zt\hat S_Z)$ follows immediately.
\begin{equation}
\cos(g_Z t\hat{S}_Z)|{+}\rangle^{\otimes N}
=
\bigotimes_{k=1}^N
\cos(g_Z t Z_k)|{+}\rangle_k,
\label{eq:factorize}
\end{equation}
and analogously for the sine term.
This factorization is the key property that makes the dynamics
analytically tractable for arbitrary $N$: each environment
qubit evolves independently, conditioned on the logical-block state.

% ============================================================
\section{Symmetric evolution of the two logical branches}
\label{sec:symmetric_branches}
% ============================================================

\subsection{Evolution of the $|\bar{z}_+\rangle_b$ branch}

Applying the exact propagator Eq.~\eqref{eq:UN} to the
initial state
$|\Psi_+(0)\rangle = |\bar{z}_+\rangle_b\otimes|{+}\rangle^{\otimes N}$
and using $\hat{Z}_b|\bar{z}_\pm\rangle_b = |\bar{z}_\mp\rangle_b$ gives
\begin{equation}
|\Psi_+(t)\rangle
=
|\bar{z}_+\rangle_b\otimes\cos(g_Z t\hat{S}_Z)|{+}\rangle^{\otimes N}
-
i\,|\bar{z}_-\rangle_b\otimes\sin(g_Z t\hat{S}_Z)|{+}\rangle^{\otimes N}.
\label{eq:psi_plus_cossin}
\end{equation}
We now introduce the two environment branch states
\begin{align}
|\phi_+(t)\rangle &= e^{-ig_Z t\hat{S}_Z}|{+}\rangle^{\otimes N},
\label{eq:phi_plus}
\\
|\phi_-(t)\rangle &= e^{+ig_Z t\hat{S}_Z}|{+}\rangle^{\otimes N},
\label{eq:phi_minus}
\end{align}
in terms of which the operator cosine and sine act as
\begin{align}
\cos(g_Z t\hat{S}_Z)|{+}\rangle^{\otimes N}
&=
\frac{|\phi_+(t)\rangle + |\phi_-(t)\rangle}{2},
\label{eq:cos_phi}
\\
\sin(g_Z t\hat{S}_Z)|{+}\rangle^{\otimes N}
&=
\frac{|\phi_-(t)\rangle - |\phi_+(t)\rangle}{2i}.
\label{eq:sin_phi}
\end{align}
Substituting Eqs.~\eqref{eq:cos_phi}--\eqref{eq:sin_phi}
into Eq.~\eqref{eq:psi_plus_cossin} and simplifying:
%\begin{align}
%|\Psi_+(t)\rangle &= |\bar{z}_+\rangle_b\otimes \frac{|\phi_+\rangle+|\phi_-\rangle}{2} - |\bar{z}_-\rangle_b\otimes \frac{|\phi_+\rangle-|\phi_-\rangle}{2}.
%\end{align}
%Collecting terms by the environment branch state yields the exact result
\begin{equation}
|\Psi_+(t)\rangle
=
\frac{|\bar{z}_+\rangle_b+|\bar{z}_-\rangle_b}{2}
\,|\phi_+(t)\rangle
+
\frac{|\bar{z}_+\rangle_b-|\bar{z}_-\rangle_b}{2}
\,|\phi_-(t)\rangle.
\label{eq:psi_plus_final}
\end{equation}
Note that $(|\bar{z}_+\rangle_b\pm|\bar{z}_-\rangle_b)/\sqrt{2}$
are proportional to the computational-basis states
$|0\rangle^{\otimes 3}$ and $|1\rangle^{\otimes 3}$
respectively, so Eq.~\eqref{eq:psi_plus_final} expresses the
evolved state as a superposition of two definite-parity
logical branches, each entangled with a distinct environment
state.

% ============================================================
\subsection{Evolution of the $|\bar{z}_-\rangle_b$ branch}
\label{sec:minusbranch}
% ============================================================

For the orthogonal initial state
$|\Psi_-(0)\rangle = |\bar{z}_-\rangle_b\otimes|{+}\rangle^{\otimes N}$,
applying $\hat{U}_N(t)$ and using
$\hat{Z}_b|\bar{z}_-\rangle_b = |\bar{z}_+\rangle_b$ gives
\begin{equation}
|\Psi_-(t)\rangle
=
|\bar{z}_-\rangle_b\otimes\cos(g_Z t\hat{S}_Z)|{+}\rangle^{\otimes N}
-
i\,|\bar{z}_+\rangle_b\otimes\sin(g_Z t\hat{S}_Z)|{+}\rangle^{\otimes N}.
\label{eq:psi_minus_cossin}
\end{equation}
Substituting the identities Eqs.~\eqref{eq:cos_phi}
and~\eqref{eq:sin_phi} and simplifying 
%$-i/(2i)=-1/2$:
%\begin{align}
%|\Psi_-(t)\rangle &= \frac{1}{2}|\bar{z}_-\rangle_b \bigl(|\phi_+(t)\rangle+|\phi_-(t)\rangle\bigr) +\frac{1}{2}|\bar{z}_+\rangle_b \bigl(|\phi_+(t)\rangle-|\phi_-(t)\rangle\bigr).
%\end{align}
%Collecting by environment branch state yields the exact result
\begin{equation}
|\Psi_-(t)\rangle
=
\frac{|\bar{z}_+\rangle_b+|\bar{z}_-\rangle_b}{2}
\,|\phi_+(t)\rangle
-
\frac{|\bar{z}_+\rangle_b-|\bar{z}_-\rangle_b}{2}
\,|\phi_-(t)\rangle.
\label{eq:psi_minus_final}
\end{equation}

\paragraph{Symmetry between the two branches.}
Comparing Eqs.~\eqref{eq:psi_plus_final}
and~\eqref{eq:psi_minus_final},
\begin{align}
|\Psi_+(t)\rangle
&=
\frac{|\bar{z}_+\rangle_b+|\bar{z}_-\rangle_b}{2}\,|\phi_+(t)\rangle
+
\frac{|\bar{z}_+\rangle_b-|\bar{z}_-\rangle_b}{2}\,|\phi_-(t)\rangle,
\label{eq:psi_plus_recall}
\\
|\Psi_-(t)\rangle
&=
\frac{|\bar{z}_+\rangle_b+|\bar{z}_-\rangle_b}{2}\,|\phi_+(t)\rangle
-
\frac{|\bar{z}_+\rangle_b-|\bar{z}_-\rangle_b}{2}\,|\phi_-(t)\rangle.
\label{eq:psi_minus_recall}
\end{align}
The two branches share identical coefficients on $|\phi_+(t)\rangle$
and differ only in the sign of the $|\phi_-(t)\rangle$ coefficient.
This perfect symmetry reflects the $\hat{Z}_b\to-\hat{Z}_b$
invariance of $\hat{H}_0$: reversing the logical branch simply
reverses the relative phase between the two environment states.
The subsequent analysis — decoherence factor, logical fidelity,
Holevo information, and Darwinistic redundancy — is governed
entirely by the overlap
$\Gamma_N(t)=\langle\phi_-(t)|\phi_+(t)\rangle$,
which is the same for both branches.

% ============================================================
\section{Decoherence factor and logical coherence}
\label{sec:decoherence}
% ============================================================

The central quantity governing all subsequent results is the
overlap between the two environment branch states,
\begin{equation}
\Gamma_N(t)
=
\langle\phi_-(t)|\phi_+(t)\rangle,
\label{eq:GammaN_def}
\end{equation}
which controls the degree of logical decoherence.
As $\Gamma_N(t)$ decreases toward zero, the environment becomes
increasingly capable of distinguishing the two $\hat Z_b$
branches, thereby suppressing their coherence. In the logical
basis this loss of coherence appears as population transfer
between $|\bar z_+\rangle$ and $|\bar z_-\rangle$, reducing the
logical fidelity while increasing the classical information
accessible from the environment. The overlap $\Gamma_N (t)$ quantifies the distinguishability of the two conditional environment states: smaller overlap corresponds to greater distinguishability and therefore greater information acquired by the environment about the logical block.
%When $\Gamma_N(t)=1$ the two environment branches are identical and the logical block remains fully coherent; as $\Gamma_N(t)$ decreases toward zero, the environment becomes increasingly capable of distinguishing the two logical states, simultaneously suppressing logical coherence, reducing QEC fidelity, and increasing the classical information accessible to environment observers — the hallmark of Darwinistic behaviour.

\subsection{Factorization and single-qubit overlap}
Since the operators $Z_k$ mutually commute,
$
e^{\mp ig_Zt\hat S_Z}
=
\prod_{k=1}^{N}
e^{\mp ig_Zt Z_k},
$
and therefore the conditional environment states factorize as
$
|\phi_\pm(t)\rangle
=
\bigotimes_{k=1}^{N}
e^{\mp ig_ZtZ_k}|+\rangle_k.
$
Consequently, their overlap factorizes as
\begin{equation}
\Gamma_N(t)
=
\prod_{k=1}^{N}
\langle\phi_-^{(k)}(t)|\phi_+^{(k)}(t)\rangle
=
\bigl[
\langle\phi_-^{(1)}(t)|\phi_+^{(1)}(t)\rangle
\bigr]^N,
\label{eq:Gamma_factorize}
\end{equation}
where the second equality follows because all environment
qubits evolve identically.
For a single environment qubit,
\begin{align}
|\phi_+^{(1)}(t)\rangle
&=
\cos(g_Z t)|{+}\rangle - i\sin(g_Z t)|{-}\rangle,
\\
|\phi_-^{(1)}(t)\rangle
&=
\cos(g_Z t)|{+}\rangle + i\sin(g_Z t)|{-}\rangle,
\end{align}
as $e^{-i\theta \hat A}
=
\cos\theta\,\mathbb I
-
i\sin\theta\,\hat A$ 
for any operator $\hat A$ satisfying
$
\hat A^2=\mathbb I
$.
So the single-qubit inner product is
\begin{equation}
\langle\phi_-^{(1)}(t)|\phi_+^{(1)}(t)\rangle
=
\cos^2(g_Z t)
-
\sin^2(g_Z t)
=
\cos(2g_Z t),
\end{equation}
where we used $\langle{\pm}|{\pm}\rangle=1$, $\langle{\pm}|{\mp}\rangle=0$.
Substituting into Eq.~\eqref{eq:Gamma_factorize} gives the
exact decoherence factor
\begin{equation}
\Gamma_N(t)
=
\cos^N(2g_Z t),
\qquad
0\le t\le\frac{\pi}{4g_Z}.
\label{eq:GammaN}
\end{equation}
$\Gamma_N(t)$ decays monotonically from $\Gamma_N(0)=1$
(no decoherence) to $\Gamma_N(\pi/4g_Z)=0$ (complete
decoherence), with the decay accelerating as $N$ increases:
each additional environment qubit multiplies the overlap by
another factor of $\cos(2g_Z t)<1$, driving $\Gamma_N$
exponentially toward zero.

% ============================================================
\subsection{Reduced logical density matrix}
\label{sec:logicaldensity}
% ============================================================
Starting from $|\Psi_+(t)\rangle$ in Eq.~\eqref{eq:psi_plus_final},
the full density matrix is
$\rho(t)=|\Psi_+(t)\rangle\langle\Psi_+(t)|$.
Introducing the shorthand $|A\rangle=|\bar{z}_+\rangle_b+|\bar{z}_-\rangle_b$
and $|B\rangle=|\bar{z}_+\rangle_b-|\bar{z}_-\rangle_b$, tracing
over the environment gives
\begin{equation}
\rho_b(t)
=
\mathrm{Tr}_E[\rho(t)]
=
\frac{1}{4}
\Bigl[
|A\rangle\langle A|
+|B\rangle\langle B|
+\Gamma_N(t)\bigl(|A\rangle\langle B|+|B\rangle\langle A|\bigr)
\Bigr],
\end{equation}
%where $\Gamma_N(t)=\langle\phi_-(t)|\phi_+(t)\rangle$ and we used
%$\mathrm{Tr}_E[|\phi_\pm\rangle\langle\phi_\mp|]=\Gamma_N^{(*)}$.
where we used
$
\mathrm{Tr}_E\!\left(|\phi_+\rangle\langle\phi_-|\right)
=
\mathrm{Tr}_E\!\left(|\phi_-\rangle\langle\phi_+|\right)
=
\Gamma_N(t),
$
Substituting $|A\rangle\langle A|+|B\rangle\langle B|
=2(|\bar{z}_+\rangle\langle\bar{z}_+|+|\bar{z}_-\rangle\langle\bar{z}_-|)$
and $|A\rangle\langle B|+|B\rangle\langle A|
=2(|\bar{z}_+\rangle\langle\bar{z}_+|-|\bar{z}_-\rangle\langle\bar{z}_-|)$
yields the exact reduced block state
\begin{equation}
\rho_b(t)
=
\frac{1+\Gamma_N(t)}{2}\,|\bar{z}_+\rangle\langle\bar{z}_+|
+
\frac{1-\Gamma_N(t)}{2}\,|\bar{z}_-\rangle\langle\bar{z}_-|,
\label{eq:rho_block}
\end{equation}
which in the logical basis $\{|\bar{z}_+\rangle,|\bar{z}_-\rangle\}$
takes the diagonal matrix form
\begin{equation}
\rho_b(t)
=
\frac{1}{2}
\begin{pmatrix}
1+\Gamma_N(t) & 0 \\ 0 & 1-\Gamma_N(t)
\end{pmatrix}.
\label{eq:rho_matrix}
\end{equation}
The reduced state is diagonal in the logical basis
$\{|\bar z_+\rangle,|\bar z_-\rangle\}$, with populations
$p_\pm(t)=[1\pm\Gamma_N(t)]/2$. Thus, although the interaction
is pure dephasing in the $\hat Z_b$ eigenbasis
$\{|000\rangle,|111\rangle\}$, it appears as population transfer
between the logical codewords in the basis used throughout this
work. At $t=0$, $\Gamma_N=1$ and
$\rho_b=|\bar z_+\rangle\langle\bar z_+|$, whereas
$\Gamma_N\to0$ yields the maximally mixed logical state
$\rho_b\to\mathbb I_b/2$.

% ============================================================
\subsection{Weak-coupling Gaussian decoherence and timescale}
\label{sec:gaussian}
% ============================================================

For $g_Z t\ll1$, the Taylor expansion
$\cos(2g_Z t)\approx 1-2g_Z^2t^2$ gives
\begin{equation}
\Gamma_N(t)
\approx
(1-2g_Z^2t^2)^N
\approx
e^{-2Ng_Z^2t^2},
\label{eq:Gamma_gaussian}
\end{equation}
where the second approximation uses $(1-\epsilon)^N\approx e^{-N\epsilon}$
for $\epsilon=2g_Z^2t^2\ll1$.
This Gaussian envelope identifies the characteristic decoherence
timescale
\begin{equation}
\tau_D
=
\frac{1}{g_Z\sqrt{2N}},
\label{eq:tauD}
\end{equation}
at which $\Gamma_N(\tau_D)=e^{-1}$.
The $N^{-1/2}$ scaling means that each additional environment
qubit accelerates decoherence: adding qubits from $N$ to $N+1$
multiplies $\Gamma_N$ by a further factor $\cos(2g_Zt)<1$,
so even a weakly coupled environment of many qubits can produce
extremely rapid dephasing.
For example, with $\cos(2g_Z t)=0.99$, one finds
$\Gamma_1=0.99$ but $\Gamma_{1000}=0.99^{1000}\approx4\times10^{-5}$:
a factor of $\sim10^4$ suppression from environment size alone.
This collective amplification of decoherence by independent
environment qubits is the microscopic mechanism underlying
both the degradation of QEC performance and the onset of
Darwinistic redundancy studied in the main text.

% ============================================================
\section{Logical fidelity under imperfect recovery}
\label{sec:FLN}
% ============================================================

We incorporate imperfect QEC recovery via the channel
\begin{equation}
\mathcal{E}_\eta(\rho)
=
P_+\rho P_+
+
\eta\,\hat{Z}_b P_-\rho P_-\hat{Z}_b
+
(1-\eta)\,P_-\rho P_-,
\label{eq:channel_SM}
\end{equation}
where $P_\pm=|\bar{z}_\pm\rangle\langle\bar{z}_\pm|$ are the
logical projectors and $\eta\in[0,1]$ is the recovery
efficiency.
From the reduced block state Eq.~\eqref{eq:rho_block},
the syndrome probabilities are
\begin{equation}
p_+(t) = \frac{1+\Gamma_N(t)}{2},
\qquad
p_-(t) = \frac{1-\Gamma_N(t)}{2}.
\end{equation}
Since the recovery succeeds with probability $\eta$
whenever the block occupies the
$|\bar z_-\rangle$ sector,
the post-recovery logical fidelity is
$F_L^{(N)} = p_+ + \eta p_-$,
which gives
\begin{equation}
F_L^{(N)}(t,\eta)
=
\frac{1+\eta+(1-\eta)\Gamma_N(t)}{2}
=
\frac{1+\eta+(1-\eta)\cos^N(2g_Z t)}{2}.
\label{eq:FL_SM}
\end{equation}
This closed-form expression is valid for all $N$,
$t$, $\eta$, and $g_Z$.
It ranges from $F_L^{(N)}=1$ at $t=0$ (perfect fidelity,
no decoherence) to $F_L^{(N)}=(1+\eta)/2=F_c$ as
$\Gamma_N\to0$ (minimum achievable fidelity).

\paragraph{Weak-coupling regime and decoherence timescale.}
Using the Gaussian approximation
$\Gamma_N(t)\approx e^{-2Ng_Z^2t^2}$
derived in Eq.~\eqref{eq:Gamma_gaussian}, the fidelity
becomes
\begin{equation}
F_L^{(N)}(t,\eta)
\approx
\frac{1+\eta+(1-\eta)e^{-2Ng_Z^2t^2}}{2},
\qquad g_Z t\ll1.
\label{eq:FL_gaussian}
\end{equation}
The fidelity decays on the timescale $\tau_D$
(Eq.~\eqref{eq:tauD}) already identified from the decoherence
factor.
The $N^{-1/2}$ scaling has a direct operational consequence:
\begin{equation}
N\uparrow
\;\Longrightarrow\;
\tau_D \propto N^{-1/2} \implies
F_L^{(N)}(t,\eta)~\text{decays more rapidly},
\label{eq:cascade}
\end{equation}
Thus, increasing the number of environment qubits simultaneously accelerates the formation of redundant classical records in the environment—the mechanism underlying Quantum Darwinism—and the degradation of the post-recovery logical fidelity, reflecting the corresponding cost to QEC performance.
Equation~\eqref{eq:FL_SM} is the starting point for
deriving the exact Darwinism--QEC tradeoff in the following
section.

% ============================================================
\section{Holevo information of environment fragments}
\label{sec:holevo_fragments}
% ============================================================

We now quantify how much classical information about the logical
state is accessible to an observer holding $m$ environment
qubits ($1\le m\le N$).
This is the key quantity for Quantum Darwinism: redundancy
arises when many independent fragments each carry sufficient
classical information to identify the logical state of the
block.

\subsection{Conditional fragment states and their overlap}

%Since the global environment states factorize (Eq.~\eqref{eq:factorize})
Since the conditional environment states
$|\phi_\pm(t)\rangle$ are product states, the conditional state of the
first $m$ qubits given logical branch $\pm$ is the pure state
\begin{equation}
|\phi_\pm^{(m)}(t)\rangle
=
\bigotimes_{k=1}^{m}
\bigl[
\cos(g_Z t)|{+}\rangle_k
\mp
i\sin(g_Z t)|{-}\rangle_k
\bigr].
\label{eq:phi_m}
\end{equation}
By the same calculation as for the full environment
(Sec.~\ref{sec:decoherence}), the overlap of the two
conditional $m$-qubit states is
\begin{equation}
\gamma_m(t)
\;=\;
\langle\phi_-^{(m)}(t)|\phi_+^{(m)}(t)\rangle
\;=\;
\cos^m(2g_Z t)
\;=\;
\bigl[\Gamma_N(t)\bigr]^{m/N},
\label{eq:gamma_m}
\end{equation}
where the last equality shows that the $m$-qubit fragment
overlap is exactly the $m/N$ power of the full decoherence
factor: each environment qubit contributes one factor of
$\cos(2g_Z t)$.

\subsection{Holevo information}

With equal prior probabilities $\tfrac{1}{2}$ for each logical
state, the Holevo information of fragment $F$ is
\begin{equation}
\chi_m(t)
=
S(\bar{\rho}_F)
-
\tfrac{1}{2}S(\rho_F^{(+)})
-
\tfrac{1}{2}S(\rho_F^{(-)}),
\label{eq:holevo_def}
\end{equation}
where $\bar{\rho}_F=\tfrac{1}{2}(\rho_F^{(+)}+\rho_F^{(-)})$
and $\rho_F^{(\pm)}=|\phi_\pm^{(m)}\rangle\langle\phi_\pm^{(m)}|$.
Since both conditional states are pure,
$S(\rho_F^{(\pm)})=0$, so $\chi_m=S(\bar{\rho}_F)$.

To find the eigenvalues of $\bar{\rho}_F$, define the
normalized symmetric and antisymmetric combinations
\begin{equation}
|s\rangle
=
\frac{|\phi_+^{(m)}\rangle+|\phi_-^{(m)}\rangle}
     {\sqrt{2(1+\gamma_m)}},
\qquad
|a\rangle
=
\frac{|\phi_+^{(m)}\rangle-|\phi_-^{(m)}\rangle}
     {\sqrt{2(1-\gamma_m)}},
\label{eq:sa_basis}
\end{equation}
which are orthonormal since
$\langle s|a\rangle=0$.
Acting with $\bar{\rho}_F$ on $|s\rangle$:
\begin{equation}
\bar{\rho}_F|s\rangle
=
\tfrac{1}{2}\bigl(
\langle\phi_+^{(m)}|s\rangle\,|\phi_+^{(m)}\rangle
+
\langle\phi_-^{(m)}|s\rangle\,|\phi_-^{(m)}\rangle
\bigr)
=
\frac{1+\gamma_m}{2}\,|s\rangle,
\end{equation}
and similarly $\bar{\rho}_F|a\rangle=\tfrac{1-\gamma_m}{2}|a\rangle$.
The two eigenvalues of $\bar{\rho}_F$ are therefore
\begin{equation}
\mu_\pm
=
\frac{1\pm\gamma_m(t)}{2}
=
\frac{1\pm\cos^m(2g_Z t)}{2},
\label{eq:eigenvalues}
\end{equation}
giving the exact Holevo information
\begin{equation}
\chi_m(t)
=
H_2\!\left(
\frac{1+\cos^m(2g_Z t)}{2}
\right),
\label{eq:holevo}
\end{equation}
where $H_2(p)=-p\ln p-(1-p)\ln(1-p)$ is the binary entropy.

\subsection{Limiting behaviour}
At $t=0$: $\gamma_m=1$, so $\mu_+=1$, $\mu_-=0$, and
$\chi_m=H_2(1)=0$ — no information is accessible from any
fragment before decoherence begins.
At $t=\pi/4g_Z$: $\gamma_m=0$, so $\mu_\pm=\tfrac{1}{2}$
and $\chi_m=H_2(\tfrac{1}{2})=\ln2$ — the fragment carries
one full bit of classical information, perfectly distinguishing
the two logical states.
For fixed $t$, $\chi_m$ is strictly increasing in $m$:
larger fragments are more informative, since
$|\cos^m(2g_Zt)|<|\cos^{m'}(2g_Zt)|$ for $m>m'$.
This monotonicity is the microscopic basis for the onset of
Darwinistic redundancy studied in the following section.

% ============================================================
\section{Darwinistic redundancy}
\label{sec:redundancy}
% ============================================================

Quantum Darwinism is characterised not merely by information
leakage to the environment, but by its \emph{redundant
proliferation} into many independent fragments.
The redundancy $R_\delta$ quantifies how many non-overlapping
fragments each carry at least a fraction $(1-\delta)$ of the
maximum classical information $\ln2$.

\subsection{Definition and threshold condition}
\label{sec:redundancy_def}

Let $m_\delta(t)$ be the smallest fragment size satisfying
$\chi_{m_\delta}=(1-\delta)\ln2$.
Substituting Eq.~\eqref{eq:holevo}:
\begin{equation}
H_2\!\left(\frac{1+\cos^{m_\delta}(2g_Z t)}{2}\right)
=
(1-\delta)\ln2.
\label{eq:threshold}
\end{equation}
Since $H_2$ is monotone decreasing on $[\tfrac{1}{2},1]$,
this equation has a unique solution in terms of the
\emph{critical overlap}
\begin{equation}
\gamma^*(\delta)
=
2H_2^{-1}\!\bigl[(1-\delta)\ln2\bigr]-1,
\label{eq:gamma_star}
\end{equation}
where $H_2^{-1}$ is the upper-branch inverse
($p\in(\tfrac{1}{2},1)$).
Equation~\eqref{eq:threshold} then becomes
$\cos^{m_\delta}(2g_Z t)=\gamma^*(\delta)$,
with $\gamma^*\in(0,1)$ depending only on $\delta$.
The Darwinistic redundancy is
\begin{equation}
R_\delta(t) = \frac{N}{m_\delta(t)}.
\label{eq:Rdelta_def}
\end{equation}

\subsection{Closed-form redundancy}
\label{sec:exact_redundancy}

Taking the logarithm of $\cos^{m_\delta}(2g_Z t)=\gamma^*$
and solving for $m_\delta$ gives
\begin{equation}
m_\delta(t)
=
\frac{\ln\gamma^*(\delta)}{\ln[\cos(2g_Z t)]},
\label{eq:mdelta}
\end{equation}
where both numerator and denominator are negative for
$t\in(0,\pi/4g_Z)$, ensuring $m_\delta>0$.
Substituting into Eq.~\eqref{eq:Rdelta_def} yields the
exact result
\begin{equation}
R_\delta(t)
=
N\,
\frac{\ln[\cos(2g_Z t)]}{\ln\gamma^*(\delta)},
\qquad
t\in\Bigl(0,\,\frac{\pi}{4g_Z}\Bigr).
\label{eq:Rdelta_exact}
\end{equation}
Both logarithms are negative so $R_\delta>0$.
As $t\to0$, $R_\delta\to0$ (no coupling, no redundancy);
as $t\to\pi/4g_Z$, $\cos(2g_Zt)\to0$ and
$R_\delta\to+\infty$ (maximum decoherence, unbounded
redundancy for continuous $m_\delta$, cut off at $N$ for
finite environments).

\subsection{Weak-coupling limit}
\label{sec:weak_coupling}

For $g_Z t\ll1$, the expansion
$\ln[\cos(2g_Z t)]\simeq-2g_Z^2t^2$ gives
\begin{equation}
m_\delta(t)\simeq\frac{C(\delta)}{g_Z^2t^2},
\qquad
R_\delta(t)\simeq\frac{N g_Z^2t^2}{C(\delta)},
\label{eq:Rdelta_weak}
\end{equation}
where
$
C(\delta) = \frac{-\ln\gamma^*(\delta)}{2}
\label{eq:C_delta}
$
is a positive, $\delta$-dependent constant.
\noindent
The scaling $R_\delta\propto Ng_Z^2t^2$ shows that redundancy
grows linearly in $N$ (more environment qubits provide more
independent copies of the logical-state information) and
quadratically in the interaction parameter $g_Zt$ (stronger
coupling imprints more information per qubit).
The same parameter $g_Zt$ drives both the increase in
redundancy and the decrease in logical fidelity
(Eq.~\eqref{eq:FL_SM}), making the Darwinism--QEC competition
directly visible at the level of the microscopic coupling.

% ============================================================
\section{Exact Darwinism--QEC tradeoff}
\label{sec:tradeoff}
% ============================================================

The logical fidelity Eq.~\eqref{eq:FL_SM} and the redundancy
Eq.~\eqref{eq:Rdelta_exact} both depend on a single
microscopic parameter: the decoherence factor $\Gamma_N(t)
=\cos^N(2g_Zt)$.
Eliminating $\Gamma_N$ between them yields a direct,
parameter-free relation.
\\
\\
From Eq.~\eqref{eq:FL_SM},
\begin{equation}
\Gamma_N(t)
=
\frac{2F_L^{(N)}-(1+\eta)}{1-\eta},
\label{eq:Gamma_from_FL_SM}
\end{equation}
so
%\begin{equation}
%\cos(2g_Zt) = \left[ \frac{2F_L^{(N)}-(1+\eta)}{1-\eta} \right]^{1/N},
%\label{eq:cos_from_FL}
%\end{equation}
%and therefore
\begin{equation}
\ln[\cos(2g_Zt)]
=
\frac{1}{N}
\ln\!\left[
\frac{2F_L^{(N)}-(1+\eta)}{1-\eta}
\right].
\label{eq:log_cos}
\end{equation}
Substituting into Eq.~\eqref{eq:Rdelta_exact},
the $N$ in the numerator and the $1/N$ in
Eq.~\eqref{eq:log_cos} cancel exactly, giving
\begin{equation}
R_\delta
=
\frac{1}{\ln\gamma^*(\delta)}
\ln\!\left[
\frac{2F_L^{(N)}-(1+\eta)}{1-\eta}
\right].
\label{eq:tradeoff_SM}
\end{equation}
This is the central result: Eq.~\eqref{eq:tradeoff_SM}
is exact, contains no approximation, and is independent of
$t$, $g_Z$, and $N$.
The cancellation of $N$ reflects the fact that both
$F_L^{(N)}$ and $R_\delta$ are determined by $\Gamma_N$
alone: knowing the post-recovery fidelity uniquely fixes
the decoherence factor and hence the entire Darwinistic
information structure of the environment.
\\
\\
%\subsection{Properties of the tradeoff}
\paragraph{Non-negativity.}
Since $F_L^{(N)}\in[(1+\eta)/2,\,1]$, the argument of the
logarithm lies in $(0,1]$, so $\ln[\cdots]\le0$.
Combined with $\ln\gamma^*<0$, this gives $R_\delta\ge0$.

\paragraph{Strict monotone competition.}
Differentiating Eq.~\eqref{eq:tradeoff_SM}
with respect to $F_L^{(N)}$,
\begin{equation}
\frac{dR_\delta}{dF_L^{(N)}}
=
\frac{1}{\ln\gamma^*(\delta)}
\cdot
\frac{2}{2F_L^{(N)}-(1+\eta)}
< 0,
\label{eq:dR_dFL_SM}
\end{equation}
since both factors are negative.
Redundancy is therefore a strictly decreasing function of
logical fidelity: every improvement in QEC performance
reduces the Darwinistic redundancy and vice versa.

\paragraph{Boundary values.}
At $F_L^{(N)}=1$ (perfect QEC): the argument equals
$(1-\eta)/(1-\eta)=1$, so $\ln(\cdots)=0$ and $R_\delta=0$
— no classical records proliferate in the environment.
As $F_L^{(N)}\to F_c^+$ where $F_c=(1+\eta)/2$: the argument
$\to0^+$, so $\ln(\cdots)\to-\infty$ and
$R_\delta\to+\infty$ — unbounded redundancy at maximum
decoherence, cut off at $R_\delta=N$ for finite environments.

\subsection{Critical scaling near $F_c$}

Writing $\varepsilon=F_L^{(N)}-F_c$,
Eq.~\eqref{eq:tradeoff_SM} becomes
\begin{equation}
R_\delta
=
\frac{1}{\ln\gamma^*(\delta)}
\ln\!\left[\frac{2\varepsilon}{1-\eta}\right]
=
-\frac{1}{|\ln\gamma^*(\delta)|}
\ln\varepsilon
+
\mathrm{const},
\label{eq:tradeoff_eps}
\end{equation}
where $\mathrm{const}=-\ln[2/(1-\eta)]/|\ln\gamma^*|$
depends only on $\eta$ and $\delta$.
As $\varepsilon\to0^+$, $R_\delta$ diverges logarithmically,
\begin{equation}
R_\delta \sim \frac{|\ln\varepsilon|}{|\ln\gamma^*(\delta)|}
\quad
\text{as}
\quad
F_L^{(N)}\to F_c^+,
\label{eq:log_div}
\end{equation}
and accordingly,
\begin{equation}
\left|\frac{dR_\delta}{dF_L^{(N)}}\right|
=
\frac{1}{|\ln\gamma^*(\delta)|\,\varepsilon}
\sim
\varepsilon^{-1}
\label{eq:susceptibility_SM}
\end{equation}
diverges as a power law, quantifying the critical sensitivity
of the Darwinism--QEC competition near the onset of complete
decoherence.

% ============================================================
\section{Logical fidelity determines the Darwinism spectrum}
\label{sec:darwinism_spectrum}
% ============================================================

The exact tradeoff Eq.~\eqref{eq:tradeoff_SM} relates the
logical fidelity to the Darwinistic redundancy $R_\delta$,
which depends only on the threshold-crossing fragment size
$m_\delta$.
A strictly stronger result holds: the logical fidelity
uniquely determines the Holevo information $\chi_m$ of
\emph{every} environment fragment, for all $m=1,\ldots,N$.

Within the present model, eliminating $\Gamma_N(t)$ from
Eq.~\eqref{eq:holevo} using Eq.~\eqref{eq:Gamma_from_FL_SM}
and the identity $\cos^m(2g_Zt)=\Gamma_N^{m/N}(t)$
(valid for $t\in[0,\pi/4g_Z]$) yields the exact result:

For every fragment of $m$ environment qubits,
\begin{equation}
\chi_m
=
H_2\!\left[
\frac{1}{2}
\left(
1
+
\left(
\frac{2F_L^{(N)}-(1+\eta)}{1-\eta}
\right)^{m/N}
\right)
\right].
\label{eq:chi_from_FL}
\end{equation}

\noindent
Equation~\eqref{eq:chi_from_FL} contains no explicit
dependence on the interaction time or coupling strength.
Once the logical fidelity and the fragment fraction
$m/N$ are specified, the Holevo information of every
environment fragment is uniquely determined without
reference to the microscopic dynamics.
This is strictly stronger than the redundancy tradeoff, which
determines only the threshold crossing at $m=m_\delta$; here
the entire partial-information curve $\chi_m$ for all
$m$ is fixed by $F_L^{(N)}$ alone.
\\
\\
Since
$
\chi_m
=
H_2\!\left[
\frac{1+\Gamma_N^{m/N}}{2}
\right],
$
and $H_2((1+x)/2)$ is monotonically decreasing on
$x\in[0,1]$, the Holevo information increases monotonically
with fragment size,
$
\chi_1\le\chi_2\le\cdots\le\chi_N.
$
Thus progressively larger environment fragments reveal
increasingly more classical information about the logical
state, with the full environment attaining the maximum
accessible information.

\subsection{Full-environment Holevo information}

For the complete environment ($m=N$), the exponent $m/N=1$
and Eq.~\eqref{eq:chi_from_FL} simplifies to
\begin{align}
\chi_N
&=
H_2\!\left[
\frac{1}{2}
\left(
1+\frac{2F_L^{(N)}-(1+\eta)}{1-\eta}
\right)
\right]
=
H_2\!\left(
\frac{F_L^{(N)}-\eta}{1-\eta}
\right),
\label{eq:chiN_simp}
\end{align}
Equation~\eqref{eq:chiN_simp} establishes a direct
observable-to-observable connection: measuring the
post-recovery logical fidelity immediately fixes the total
classical information encoded in the full environment.
Together with Eq.~\eqref{eq:chi_from_FL}, this shows that a
measurement of the post-recovery logical fidelity
reconstructs the complete Holevo-information curve
$\{\chi_m\}_{m=1}^{N}$, establishing logical-fidelity
measurements and environment-fragment readout as
operationally equivalent probes of the same decoherence
process.

% ============================================================
\section{Exact information balance}
\label{sec:information_balance}
% ============================================================

The exact solution reveals a precise relationship between
the entropy generated in the logical block and the classical
information acquired by the complete environment.

From Eq.~\eqref{eq:rho_block}, the reduced block state is
diagonal in the logical basis with eigenvalues
$\lambda_\pm=(1\pm\Gamma_N)/2$, giving the von Neumann entropy
\begin{equation}
S(\rho_b)
=
H_2\!\left(\frac{1+\Gamma_N(t)}{2}\right).
\label{eq:S_rho_b}
\end{equation}
The full-environment Holevo information
(Eq.~\eqref{eq:holevo} with $m=N$) is
\begin{equation}
\chi_N
=
H_2\!\left(\frac{1+\Gamma_N(t)}{2}\right).
\label{eq:chiN_recall}
\end{equation}
Comparing Eqs.~\eqref{eq:S_rho_b} and~\eqref{eq:chiN_recall}
gives the exact information balance
\begin{equation}
S(\rho_b) = \chi_N.
\label{eq:info_balance}
\end{equation}
Every nat of von Neumann entropy generated in the logical
block is exactly accounted for by the Holevo information
stored in the full environment — no information is hidden
in quantum coherences between the conditional environment
states.

This equality holds because the two contributing factors
are both determined by $\Gamma_N(t)$ through the same
function $H_2((1+\Gamma_N)/2)$.
Microscopically, the two conditions responsible are:
(i)~the conditional environment states $|\phi_\pm(t)\rangle$
are pure (product states), so $S(\rho_E^{(\pm)})=0$ and
$\chi_N=S(\bar{\rho}_E)$; and
(ii)~$S(\bar{\rho}_E)=S(\rho_b)$, which follows from the
eigenvalue structure shared by $\bar{\rho}_E$ and $\rho_b$
through the common decoherence factor $\Gamma_N$.

\subsection{Connection with logical fidelity}

Substituting $\Gamma_N=(2F_L^{(N)}-(1+\eta))/(1-\eta)$
from Eq.~\eqref{eq:Gamma_from_FL_SM}, both sides of
Eq.~\eqref{eq:info_balance} become
\begin{equation}
S(\rho_b)
=
\chi_N
=
H_2\!\left(\frac{F_L^{(N)}-\eta}{1-\eta}\right).
\label{eq:balance_FL}
\end{equation}
The logical fidelity therefore simultaneously determines
the block entropy, the full-environment classical
information, and — through Eq.~\eqref{eq:chi_from_FL} —
the Holevo information of every environment fragment.
All three are governed by the single decoherence parameter
$\Gamma_N(t)$, which is in turn uniquely fixed by the
observable $F_L^{(N)}$.

% ============================================================
\section{General fragment-information bound}
\label{sec:general_bound}
% ============================================================

Before proving the model-independent no-go theorem, we derive
a general upper bound on the classical information accessible
from an arbitrary environment fragment. Unlike the exact
relations obtained for the solvable model, the following
bound relies only on general information-theoretic
properties and the purity of the joint block--environment
state.

\paragraph{Model-independent part.}
For any bipartite state and any subsystem $F\subseteq E$,
the partial trace $\mathcal{T}_F=\mathrm{Tr}_{E\setminus F}$
is a completely positive trace-preserving (CPTP) map.
The data-processing inequality for Holevo information gives
$
\chi(F) \le \chi(E).
$
The Holevo bound further gives $\chi(E)\le S(\rho_E)$,
and for a pure joint state $|\Psi\rangle_{bE}$ the
purification identity gives $S(\rho_E)=S(\rho_b)$.
Chaining these three inequalities yields the
\emph{model-independent} bound
\begin{equation}
\chi(F)
\le
\chi(E)
\le
S(\rho_b),
\label{eq:general_bound}
\end{equation}
which holds for any Hamiltonian, any $N$, and any initial
state, requiring only that the joint block--environment
state is pure.

\paragraph{Model-specific saturation.}
In the present solvable model, the conditional environment
states are pure, so
$
\chi(E)=S(\bar\rho_E).
$
Since $\bar\rho_E$ and $\rho_b$ share the same nonzero
eigenvalue spectrum,
$
\chi(E)=S(\rho_b),
$
showing that the general upper bound is saturated.
\\
\\
For fragments consisting of $m$ environment qubits,
Eq.~\eqref{eq:general_bound} becomes
\[
\chi_m\le S(\rho_b).
\]
Combining this general inequality with the entropy bound
$S(\rho_b)\le H_2(F_{\mathrm{bare}})$ proved in the main
text immediately yields
\[
\chi_m
\le
S(\rho_b)
\le
H_2(F_{\mathrm{bare}}),
\]
which forms the key ingredient of the model-independent
no-go theorem.

\end{widetext}
\end{document}